\documentclass[conference,a4paper]{IEEEtran}

\usepackage{amsmath,amsthm,amssymb}
\usepackage{color}
\usepackage{graphicx} 
\usepackage{mathrsfs} 
\usepackage{psfrag,epsfig}
\usepackage{tikz} 
\usetikzlibrary{patterns,shapes}
\usepackage{multicol}
\usepackage[font=small,labelfont=bf]{caption}
\usepackage[font=footnotesize]{subfig}

\graphicspath{{./}{./figs/}}
\newcommand{\area}{\mathscr{A}}

\newcommand{\an}{a_n}

\renewcommand{\Re}{\mathbb{R}}
\newcommand{\prob}[1]{\mathsf{Pr}\left( #1 \right)}
\newcommand{\graph}[1]{G\left( #1 \right)}

\newcommand{\exponent}[1]{\exp \left( #1\right)}

\newcommand{\ignore}[1]{}
\theoremstyle{definition}
\newtheorem{theorem}{Theorem} 
\newtheorem{lemma}{Lemma}
\newtheorem{remark}{Remark}

\title{On Connectivity Thresholds in the Intersection of Random Key
  Graphs on Random Geometric Graphs}

\author{ \IEEEauthorblockN{B.~Santhana~Krishnan}
  \IEEEauthorblockA{Electrical Engg. Department\\
    IIT Bombay, India\\
    Email: \texttt{skrishna@ee.iitb.ac.in} \thanks{This material is
      based upon work supported by the Bharti Centre for
      Communication, EE Department, IIT Bombay.} }  \and
  \IEEEauthorblockN{Ayalvadi~Ganesh}
  \IEEEauthorblockA{Department of Mathematics,\\
    University of Bristol, United Kingdom\\
    Email: \texttt{aganesh@bristol.ac.uk}} \and
  \IEEEauthorblockN{D.~Manjunath}
  \IEEEauthorblockA{Electrical Engg. Department\\
    IIT Bombay, India\\
    Email: \texttt{dmanju@ee.iitb.ac.in} } }

\IEEEoverridecommandlockouts

\begin{document}
\sloppy

\maketitle

\begin{abstract}
  In a random key graph (RKG) of $n$ nodes each node is randomly
  assigned a key ring of $K_n$ cryptographic keys from a pool of $P_n$
  keys. Two nodes can communicate directly if they have at least one
  common key in their key rings. We assume that the $n$ nodes are
  distributed uniformly in $[0,1]^2.$ In addition to the common key
  requirement, we require two nodes to also be within $r_n$ of each
  other to be able to have a direct edge. Thus we have a random graph
  in which the RKG is superposed on the familiar random geometric
  graph (RGG). For such a random graph, we obtain tight bounds on the
  relation between $K_n,$ $P_n$ and $r_n$ for the graph to be
  asymptotically almost surely connected.
\end{abstract}

\section{Introduction}
\label{sec:intro}
Several constructions for random graphs have been proposed with
different, suitably parametrised, rules to determine the existence of
an edge between two nodes. The most well known of these are the
Erd\H{o}s-R\'{e}nyi (ER) random graphs that have independent edges;
\cite{Bollobas01} is an excellent introduction to the study of such
graphs. Most other random graphs have edges that are not
independent. An important example of the latter kind is the random
geometric graph (RGG), motivated by, among other systems, wireless
networks. Here the nodes are randomly distributed in a Euclidean space
and there is an edge between two nodes if the Euclidean distance
between them is below a specified threshold; \cite{Penrose03} provides
a comprehensive treatment of such graphs. A more recent example of a
random graph with non independent edges is the random key graph (RKG)
\cite{OYagan12}. Here there is a key pool of size $P$ and each node
randomly chooses $K$ of these for its key ring uniformly i.i.d. Two
nodes have an edge if they have at least one common key in their key
rings. Such networks have also been investigated as uniform random
intersection graphs; see e.g.,
\cite{Dia:conn:rnd:int:graph:Rybarczyk}. That the edges are not
independent in RGGs and RKGs is evident.

Recently, there is interest in random graphs in which an edge is
determined by more than one random property, i.e., \emph{intersection}
of different random graphs. The intersection of ER random graphs and
RGGs has been of interest for quite some time now. A general form of
such graphs is as follows. $n$ nodes are distributed uniformly in an
area and the probability that two nodes are connected is a function of
their distance and is independent of other edges. This has also been
called the random connection model. Recent work on such random graphs
are in \cite{Mao:Connectivity:understanding}
\ignore{Mao:Connectivity:unique:large, Mao:Connectivity:isolated,}
where connectivity properties are analyzed. In
\cite{OYagan:secure:pn}, the superposition of an ER random graph on an
RKG is considered. The construction of such a graph is as follows: an
RKG is first formed based on the key-distribution and each edge in
this graph is deleted with a specified probability.

In this paper, our interest is in the intersection of RKGs and
RGGs. $n$ nodes are distributed in a finite Euclidean space and an RGG
is formed with edges between nodes that are within $r_n$ of each
other. The network has a pool of $P_n$ keys and each node
independently chooses for itself a key ring of size $K_n.$ Each edge
of this RGG is retained if the two nodes have at least one common key
in their key rings. A more formal definition of this graph will be
provided in the next section.

An important distinction between the random graph that we consider in
this paper and the ones in \cite{Mao:Connectivity:understanding,
  OYagan:secure:pn} \ignore{Mao:Connectivity:unique:large,
  Mao:Connectivity:isolated} is that both the RKG and the RGG have
non independent edges. This complicates the analysis
significantly. The rest of the paper is organized as follows. In the
next section we formally describe the model and then provide an
overview of the literature. In Section~\ref{sec:main-result} we state
the main result and a sketch of the proof. The formal proof is in
Section~\ref{sec:proof}. We conclude in Section~\ref{sec:conclusion}.

\section{Preliminaries}

The $n$ nodes are uniformly distributed in $\area:=[0,1]^2.$ Let $x_i
\in \area$ be the location of node $i.$ A key pool with $P_n$
cryptographic keys is designated for the network of $n$ nodes. Node
$i$ chooses a random subset $S_i$ of keys from the key pool with
$|S_i|=K_n.$ Our interest is in the random graph $\graph{P_n, K_n,
  r_n}$ with $n$ nodes and edges formed as follows. An edge $(i,j),$
between $x_i, x_j \in \area,$ is present in $\graph{P_n, K_n, r_n}$ if
both of the following two conditions are satisfied.
\begin{enumerate}
\item[$E_1:$] $\Vert x_i - x_j\Vert \leq r_n$
\item[$E_2:$] $S_i \cap S_j \neq \emptyset$
\end{enumerate}
Condition $E_1$ produces a random geometric graph with cutoff $r_n.$
Imposing condition $E_2$ on $E_1$ retains the edges of the random
geometric graph for which the two nodes have a common key. Thus
$\graph{P_n,K_n,r_n}$ is a RKG-RGG.

$\graph{r_n}$ will refer to a random geometric graph in which an edge
$(i,j)$ is determined only by $E_1.$ Similarly, $\graph{P_n, K_n}$
will refer to the RKG where an edge $(i,j)$ is determined only by
$E_2.$ The following is known about the connectivity of these types of
random graphs.
\begin{theorem}\cite[Theorems 2.1, 3.2]{GuptaKumar:connectivity}
  In $\graph{r_n},$ let $\pi r_n^2 = \frac{\log n + c_n}{n}.$ Then
  \begin{eqnarray*}
    &&\hspace{-0.2in} \lim \inf_{n \to \infty} \prob{\graph{r_n} \mbox{
        is disconnected}} \ \geq \ e^{-c} \left(1 - e^{-c}
    \right) \\
    && \hspace{1.2in} \mbox{if $\lim_{n \to \infty}c_n = c$ and $0 < c <
      \infty,$}\\ 
    &&\lim_{n \to \infty} \prob{\graph{r_n} \mbox{
        is connected}} \ = \ 1 \\
    && \hspace{1.8in} \mbox{if and only if $c_n \to +\infty.$}
  \end{eqnarray*}
  This theorem is also available from~\cite[Theorem
  2]{Penrose:Longest}.
  \label{thm:RGG}
\end{theorem}
\begin{theorem}
  \cite[Theorem 4.1]{OYagan12} In $\graph{P_n, K_n},$ let $K_n
  \geq 2$ and $\frac{K_n^2}{P_n} = \frac{\log n + c_n}{n}.$ Then,
  \begin{eqnarray*}
    && \lim_{n \to \infty} \prob{\graph{P_n, K_n} \mbox{ is
        connected}} \ = \ 0 \\
    && \hspace{1.8in} \mbox{if $\lim_{n \to \infty} c_n = -\infty,$ } \\
    && \lim_{n \to \infty} \prob{\graph{P_n, K_n} \mbox{ is
        connected}} \ = \ 1 \\
    && \mbox{for $\sigma > 0,$ if $K_n \to \infty,$ $\
      P_n \geq \sigma n$ \& $\lim_{n \to \infty} c_n = \infty.$}
  \end{eqnarray*}
  \label{thm:RKG}
\end{theorem}
If $r_n=\sqrt{2}$ we see that $\graph{P_n, K_n, r_n}$ is a RKG
$\graph{P_n, K_n}$ and Theorem~\ref{thm:RKG} applies. In fact it is
easy to argue that if $r_n = r > 0,$ then Theorem~\ref{thm:RKG}
applies. Further note that if the condition for Theorem~\ref{thm:RGG}
is satisfied with $c_n \to \infty$ and $c_n \in \Theta(\log \log n)$
then the minimum degree in $\graph{r_n}$ will be a constant. This
means that if an RKG is now superposed on this, the graph will be
disconnected with a constant probability if the probability that two
nodes share a key is less than 1. Thus we will need $c_n$ to be such
that the minimum degree in $\graph{r_n}$ is unbounded; we assume $n
\pi r_n^2 = d_n,$ where $d_n \in \omega(\log n),$ and $d_n \in o(n).$

\section{Main Result}
\label{sec:main-result}

The main result of this paper is the following theorem that
characterizes the probability of connectivity of an RKG-RGG
intersection random graph.
\begin{theorem}
  Let $K_n \geq 2,$ $K_n, P_n \to \infty,$ $K_n^2/P_n \to 0,$ $P_n
  \geq 2 K_n$ and $P_n \geq \sigma n r_n^2$ where $\sigma > 0$ is a
  constant. Then
  \begin{enumerate}
  \item \label{proof:keys:necessary} If $\pi r_n^2 \frac{K_n^2}{P_n} =
    \frac{\log n + c_1}{n} \ $ with $0 < c_1 < \infty$ then
    \begin{displaymath}
      \lim_{n\to\infty} \prob{\graph{P_n, K_n, r_n} \mbox{ is
          disconnected}} \geq \frac{e^{-c_1}}{4}.
    \end{displaymath}
  \item \label{proof:keys:sufficiency} If $\pi r_n^2 \frac{K_n^2}{P_n}
    \ > \ \frac{2 \pi}{1-\delta} \frac{\log n}{n}$ for any $\delta,$
    $0 < \delta < 1,$ then for some $c_3 > 0$ and some $c_2,$ $0 < c_2
    < \infty,$
    \begin{displaymath}
      \lim_{n\to\infty} \prob{\graph{P_n, K_n, r_n} \mbox{ is
          connected}} \ \geq \ 1 - \frac{c_2}{n^{c_3}}.
    \end{displaymath}
     Thus $ \prob{\graph{P_n,
        K_n, r_n} \mbox{ is connected}} \to 1. $
  \end{enumerate}
  \label{thm:main}
\end{theorem}
The first statement of the theorem is proved in the usual way by
considering the probability of finding at least one isolated node in
the network for a specified $(P_n, K_n, r_n).$ The second part takes a
slightly different approach. We divide $\area$ into smaller square
cells whose lengths are proportional to $r_n.$ We then consider a set
of overlapping tessellations where a cell in one tessellation overlaps
with four cells in the other tessellation. Connectivity of
$\graph{P_n, K_n, r_n}$ is ensured as follows: (1)~all cells are
dense, i.e., all cells have $\Theta(n\ r_n^2)$ nodes inside them, and
(2)~the nodes in each cell form a connected subgraph. The
tessellations are illustrated in Fig.~\ref{fig:cell-tessellation}. The
proof will identify the $(P_n, K_n, r_n)$ that achieves both of these
properties.

\section{Proof of Theorem~\ref{thm:main}}
\label{sec:proof}

We will repeatedly use the following inequality. For any $0 < x <1,$
and any positive integer $n,$
\begin{equation}
  \exponent{- \frac{n x}{1 - x} } \ < \ (1 - x)^n \ < \ \exponent{ -
    n x}. \label{eq:braveinequality}
\end{equation}
See Appendix~\ref{app:sec:prelims} for details.

Also, we will be using the following lemma from \cite{OYagan12}. 
\begin{lemma}
  \label{lemma:equiv:kn2pn}
  If $\lim_{n \to \infty} \frac{K_n^2 }{P_n} = 0,$ then
  \begin{displaymath}
    \beta_n := 1 - \frac{{P_n - K_n \choose K_n}}{{P_n \choose K_n}}
    \sim \frac{K_n^2}{P_n}.
  \end{displaymath}
\end{lemma}
$\beta_n$ is the probability that two nodes share a key.

\subsection{Proof of Statement~\ref{proof:keys:necessary} of
  Theorem~\ref{thm:main}}

Let $Z_i$ denote the event that node $i, 1 \leq i \leq n,$ is
isolated, and define $\an := \pi r_n^2,$ $\beta_n := 1 - \left( {P_n -
    K_n \choose K_n} / {P_n \choose K_n} \right).$ Observe that
$\beta_n$ is the probability that two nodes have at least one common
key. From Bonferroni inequalities and symmetry,
\begin{eqnarray}
  \prob{\bigcup_{i=1}^n Z_i} & \geq & \sum_{i=1}^n \prob{Z_i} -
  \sum_{1 \leq i < j \leq n} \prob{Z_i \cap Z_j}. \nonumber \\
  & = & n \prob{Z_1} - \binom{n}{2} \prob{Z_1 \cap Z_2}
  \label{eq:lower:bound}
\end{eqnarray}
Clearly,
\begin{displaymath}
  \prob{Z_1} = \left(1 - \an \beta_n \right)^{n-1}.
\end{displaymath}
Let $\an \beta_n = (\log n + c_1) / n,$ with $0 < c_1 < \infty.$ Using
\eqref{eq:braveinequality}, we can show that
\begin{equation}
  n \prob{Z_1} \geq \exponent{-c_1} \ \exponent{- \frac{(\log n + c_1
      )^2}{n - (\log n + c_1) }}.  \label{eq:necessary:1}
\end{equation}
The details are in Appendix~\ref{app:eq:necessary:1}.

Consider two circles of radius $r_n$ centered at $x_1$ and $x_2.$ Let
$B_3$ be the intersection of the two circles, $B_1$ (resp. $B_2$) be
the part of the circle at $x_1$ (resp. $x_2$) excluding $B_{3}$ and
$B_4 := \area \setminus (B_1 \cup B_2).$ Let $d:=\Vert x_1 -
x_2\Vert.$ The areas of the regions $B_i$ depend on $d$ and we will
use $B_i$ to also to refer to the areas. Further, let $n_i$ be the
number of nodes in $B_i$ for $1 \leq i \leq 4.$ Ignoring the edge
effects, when $(n-2)$ nodes are distributed uniformly in $\area$ the
$n_i$ form a multinomial distribution with probabilities equal to
$B_i.$ We consider the following three cases as shown in
Fig.~\ref{fig:circles:case1}, ~\ref{fig:circles:case2} and
~\ref{fig:circles:case3}.
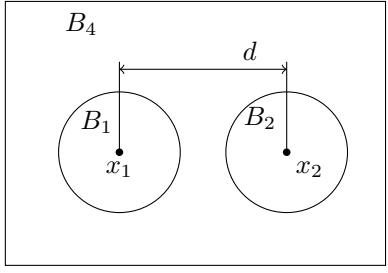
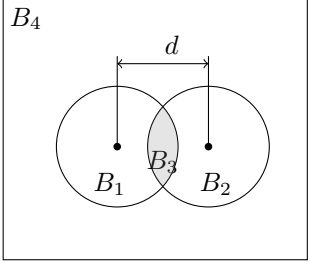
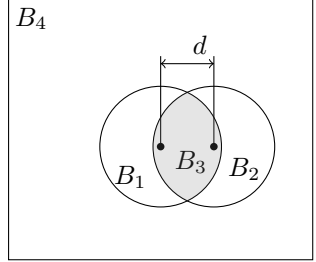
\begin{figure}
  \begin{center}
    \subfloat[Areas $B_1, B_2, B_4$ corresponding to case 1: $r_n \geq
    2 r_n.$\label{fig:circles:case1}]{
      \begin{tikzpicture}
        \draw (-3,-1.5) rectangle (2,2);
        \draw (-1.5,0) circle (0.8cm);
        \draw (0.7,0) circle (0.8cm);
        \fill [color=black] (-1.5,0) circle(.05cm);
        \fill [color=black] (0.7,0) circle(.05cm);
        \node at (-1.8, 0.4) {$B_1$};
        \node at (0.35, 0.45) {$B_2$};
        \node at (-2, 1.7) {$B_4$};
        \node at (-1.5, 0) [below] {$x_1$};
        \node at (1, 0) [below] {$x_2$};
        \draw (-1.5,0) -- (-1.5,1.2);
        \draw (0.7,0) -- (0.7,1.2);
        \draw[<->] (-1.5, 1.1) -- (0.7, 1.1);
        \node at (0,1.1) [above right] {$d$};
      \end{tikzpicture}
    }\hspace{2in}
    \subfloat[Areas $B_1, B_2, B_3, B_4$ corresponding to case 2: $d <
    r_n \leq 2 r_n.$\label{fig:circles:case2}]{
      \begin{tikzpicture}
        \draw (-2,-1.5) rectangle (2,2);
        \draw (-0.5,0) circle (0.8cm);
        \draw (0.7,0) circle (0.8cm);
        \fill [color=black] (-0.5,0) circle(.05cm);
        \fill [color=black] (0.7,0) circle(.05cm);
        \node at (-0.6, -0.5) {$B_1$};
        \node at (0.8, -0.5) {$B_2$};
        \node at (0.1, -0.2) {$B_3$};
        \node at (-1.7, 1.7) {$B_4$};
        \draw (-0.5,0) -- (-0.5,1.2);
        \draw (0.7,0) -- (0.7,1.2);
        \draw[<->] (-0.5, 1.1) -- (0.7, 1.1);
        \node at (0,1.1) [above right] {$d$};
        \begin{scope}
          \clip (-0.5,0) circle (0.8cm);
          \clip (0.7,0) circle (0.8cm);
          \fill[color=gray, opacity = .2] (-2,1.5) rectangle (2,-1.5);
        \end{scope}
      \end{tikzpicture}
    } \hspace{.05in}
    \subfloat[Areas $B_1, B_2, B_3, B_4$ corresponding to case 3: $d
    \leq r_n$\label{fig:circles:case3}]{
      \begin{tikzpicture}
        \draw (-2, -1.5) rectangle (2,2);
        \draw (0,0) circle (0.8cm);
        \draw (0.7,0) circle (0.8cm);
        \fill [color=black] (0,0) circle(.05cm);
        \fill [color=black] (0.7,0) circle(.05cm);
        \node at (-0.4, -0.4) {$B_1$};
        \node at (1.1, -0.3) {$B_2$};
        \node at (0.4, -0.2) {$B_3$};
        \node at (-1.7, 1.7) {$B_4$};
        \draw (0,0) -- (0,1.2);
        \draw (0.7,0) -- (0.7,1.2);
        \draw[<->] (0, 1.1) -- (0.7, 1.1);
        \node at (0.3,1.1) [above right] {$d$};
        \begin{scope}
          \clip (0,0) circle (0.8cm);
          \clip (0.7,0) circle (0.8cm);
          \fill[color=gray, opacity = .2] (-2,1.5) rectangle (2,-1.5);
        \end{scope}
      \end{tikzpicture}
    }
  \end{center}
  \caption{Areas to be considered for Nodes-1 and 2 to be jointly
    isolated. }
  \label{fig:circles}
\end{figure}
\begin{enumerate}
\item $d > 2r_n:$ This case happens with probability $1-4\an.$ Here
  $B_3=0$ and hence $n_3=0.$ $Z_1 \cap Z_2$ is true if each of the
  $n_1$ nodes in $B_1$ do not share a key with Node~1, and each of the
  $n_2$ nodes in $B_2$ do not share a key with Node~2. Hence
  \begin{equation}
    \prob{Z_1 \cap Z_2 \vert d > 2r_n} = \left(1 - 2 \ \an \beta_n
    \right)^{n-2}   \label{eq:necessary:case1}
  \end{equation}
\item $r_n \leq d \leq 2r_n:$ This case happens with probability
  $3\an.$ In this case, for $Z_1 \cap Z_2$ to be true the $n_1$ nodes
  in $B_1$ and $n_2$ nodes in $B_2$ should be as in the previous
  case. In addition we will need that the $n_3$ nodes in $B_3$ not
  share a key with either Node~1 or Node~2.
  { \small{
      \begin{eqnarray}
        && \hspace{-.4in} \prob{Z_1 \cap Z_2 \vert r_n \leq d \leq
          2r_n} \ \leq \nonumber \\
        && \hspace{-.1in} \exponent{- (n-2) \left( 2 - \left \Vert
              \frac{\tilde{\beta}_n}{\beta_n} - 2 \right \Vert \right) \an \beta_n }
        \label{eq:necessary:case2}
      \end{eqnarray}
    } \normalsize }
  \noindent where $\tilde{\beta}_n := 1 - \left({P_n - 2 K_n \choose
      K_n} / {P_n \choose K_n} \right).$
  See~Appendix~\ref{app:eq:necessary:2} for details.

\item $d < r_n:$ This case happens with probability $\an.$ For $Z_1
  \cap Z_2$ to be true, the conditions of the previous case should be
  satisfied. In addition Nodes 1 and 2 should also not share a
  key. Identical to the second term in~\eqref{eq:necessary:case2}, we
  have
  {\small
    {
      \begin{eqnarray}
        && \hspace{-.4in} \prob{Z_1 \cap Z_2 \vert 0 \leq d \leq
          r_n} \ \leq \nonumber \\
        && \exponent{- (n-2) \left( 2 - \left \Vert
              \frac{\tilde{\beta}_n}{\beta_n} - 2 \right \Vert \right)
          \an \beta_n }
        \label{eq:necessary:case3}
      \end{eqnarray}
    }\normalsize
  }
  See~Appendix~\ref{app:eq:necessary:3} for details.
\end{enumerate}
From~\eqref{eq:necessary:case1}, \eqref{eq:necessary:case2}
and~\eqref{eq:necessary:case3} the unconditional joint probability of
two nodes being isolated is bounded as:
{ \small {
    \begin{eqnarray*}
      \prob{Z_1 \cap Z_2} & \leq & (1-4\an) \left(1 - 2 \ \an \beta_n
      \right)^{n-2} \\
      && \hspace{-.6in}+ 4\an \frac{\exponent{\log n \left[\gamma -
            \frac{c_1 \left(2 - \gamma \right)}{ \log n} + \frac{\left(4 - 2
                \gamma \right) \an \beta_n}{\log n } \right]}}{n^2}.
    \end{eqnarray*}
  } \normalsize
}
where $\gamma := \left \Vert \frac{\tilde{\beta}_n}{\beta_n} - 2
\right \Vert.$

\noindent An upper bound on ${n \choose 2} \prob{Z_1 \cap Z_2}$ is
obtained for some $\epsilon>0$ by using $\an = d_n/n$ and $\an \beta_n
= \frac{\log n + c_1}{n}$ in the preceding inequality.
{ \small{
    \begin{equation}
      {n \choose 2} \prob{Z_1 \cap Z_2} \leq \exponent{-c_1 }
      \frac{\exponent{-c_1 + \frac{4 \left(\log n + c_1\right)}{n}
        }}{2} + \frac{2}{n^{\epsilon}}.\label{eq:prob:z12}
    \end{equation}
  }\normalsize
}

\noindent
See~Appendix~\ref{app:eq:prob:z12} and Appendix~\ref{app:beta:tilde}
for details. Using~\eqref{eq:necessary:1} and~\eqref{eq:prob:z12}
in~\eqref{eq:lower:bound}, the lower bound on $\prob{\cup_{i=1}^n
  Z_i}$ is
{ \small{
    \begin{eqnarray}
      \prob{\cup_{i=1}^{n} Z_i} & \geq & \exponent{-c_1} \left(
        \exponent{- \frac{(\log n + c_1 )^2}{n - (\log n + c_1) }}
      \right. \nonumber \\
      && \hspace{.25in} \left. - \frac{\exponent{-c_1 + \frac{4
              \left(\log n + c_1\right)}{n}}}{2} - \frac{\exponent{2
            c_1}}{n^{\epsilon}} \right) \nonumber \\
      & \geq & \frac{\exponent{-c_1}}{4}. \label{eq:necessary:final}
    \end{eqnarray}
  } \normalsize
}

\noindent Combining~\eqref{eq:necessary:final} with
Lemma~\ref{lemma:equiv:kn2pn}, we have the necessary condition of
Theorem~\ref{thm:main}. $\hfill \square$

\begin{remark}
  If $\an \beta_n = \left(\log n + c_n \right)/ n$ for any $c_n \to
  \infty,$ then using the union bound, we see that asymptotically
  almost surely, there are no isolated nodes in the graph $\graph{P_n,
    K_n, r_n}.$
  \label{rem:no:isolated}
\end{remark}

\subsection{Proof of Statement~\ref{proof:keys:sufficiency} of
  Theorem~\ref{thm:main}}

We consider two overlapping tessellations on $\area$ as shown in
Fig.~\ref{fig:cell-tessellation}, call them tessellations $1$ and $2.$
In both tessellations, $\area$ is divided into square cells of size
$s_n \times s_n$ where $1/s_n$ is an integer and $r_n = \sqrt{2} s_n.$
This means that two nodes in the same cell are within communicating
range of each other. Note the overlapping structure in the cells of
the two tessellations. 

For the proof we show the following.
\begin{enumerate}
\item In each of the tessellations, every cell is dense. Specifically,
  every cell has $\Theta(n s_n^2)$ nodes w.h.p (with high
  probability).
\item W.h.p the subgraph of $G(P_n,K_n,r_n)$ induced by the nodes in a
  cell forms a single connected component. Further w.h.p, the
  subgraphs of every cell in a tessellation have this property.
\item Use the preceding results and the overlapping structure of the
  two tessellations to argue that the graph is connected w.h.p.
\end{enumerate}
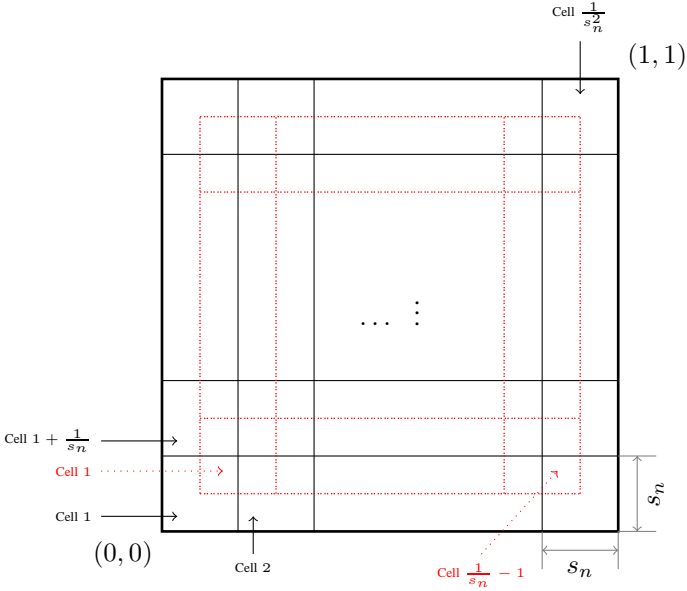
\begin{figure}
  \begin{center}
    \begin{tikzpicture}
      \draw[line width=1pt] (-3,-3) rectangle (3,3);
      \draw [-] (-2,-3) -- (-2, 3);
      \draw [-] (-1,-3) -- (-1, 3);
      \draw [-] (2,-3) -- (2, 3);
      \draw [-] (-3,-2) -- (3, -2);
      \draw [-] (-3,-1) -- (3, -1);
      \draw [-] (-3,2) -- (3, 2);
      \node at (-3, -3) [below left] {$(0,0)$};
      \node at (3, 3) [above right] {$(1,1)$};
      \draw[color=gray] (2, -3) -- (2, -3.5);
      \draw[color=gray] (3, -3) -- (3, -3.5);
      \draw[color=gray] (3, -3) -- (3.5, -3);
      \draw[color=gray] (3, -2) -- (3.5, -2);
      \draw[<->, color=gray] (2, -3.25) -- (3, -3.25);
      \draw[<->, color=gray] (3.25, -2) -- (3.25, -3);
      \node at (2.5, -3.5) {$s_n$};
      \node[rotate=90] at (3.5, -2.5) {$s_n$};
      \node at (0,0) {$\ldots \ \ \vdots$};
      \draw[color=red, line width=.5pt, dotted] (-2.5,-2.5) rectangle
      (-2.5,2.5);
      \draw[color=red, line width=.5pt, dotted] (-1.5,-2.5) rectangle
      (-1.5,2.5);
      \draw[color=red, line width=.5pt, dotted] (2.5,-2.5) rectangle
      (2.5,2.5);
      \draw[color=red, line width=.5pt, dotted] (1.5,-2.5) rectangle
      (1.5,2.5);
      \draw[color=red, line width=.5pt, dotted] (-2.5,-2.5) rectangle
      (2.5,-2.5);
      \draw[color=red, line width=.5pt, dotted] (-2.5,-1.5) rectangle
      (2.5,-1.5);
      \draw[color=red, line width=.5pt, dotted] (-2.5,1.5) rectangle
      (2.5,1.5);
      \draw[color=red, line width=.5pt, dotted] (-2.5,2.5) rectangle
      (2.5,2.5);
      \draw [->] (-3.8, -2.8) -- (-2.8, -2.8);
      \node at (-3.8, -2.8) [left] {{\tiny{Cell $1$}}\normalsize};
      \draw [->] (-3.8, -1.8) -- (-2.8, -1.8);
      \node at (-3.8, -1.8) [left] {{\tiny{Cell $1 + \frac{1}{s_n}$}}
        \normalsize};
      \draw [->] (-1.8, -3.3) -- (-1.8, -2.8);
      \node at (-1.8, -3.3) [below] {{\tiny{Cell $2$}}\normalsize};
      \draw [->] (2.5,3.5) -- (2.5,2.8);
      \node at (2.5, 3.5) [above] {{\tiny{Cell
            $\frac{1}{s_n^2}$}}\normalsize};
      \draw [color=red, dotted, ->] (-3.8, -2.2) -- (-2.2, -2.2);
      \node [color=red] at (-3.8, -2.2) [left] {{\tiny{Cell
            $1$}}\normalsize};
      \draw [color=red, dotted, ->] (1.2, -3.3) -- (2.2, -2.2);
      \node [color=red ] at (1.2, -3.3) [below] {{\tiny{Cell
            $\frac{1}{s_n} - 1$}} \normalsize};
    \end{tikzpicture}
    \caption{Tessellation of $[0,1]^2$ (with cell numbers given inside
      the cells). Tessellation $1(2)$ is shown using
      continuous(dotted) line divisions.}
    \label{fig:cell-tessellation}
  \end{center}
\end{figure}

First, we analyse denseness of each cell. Recall that $n \an = d_n,$
where $d_n \in \omega( \log n)$ and $d_n \in o(n).$ Let $N_i$ denote
the number of nodes in cell $i, \ 1 \leq i \leq 1/s_n^2.$ Clearly
$N_i$ is a binomial random variable with parameters $(n, s_n^2).$ Let
$W_i$ indicate the event that cell $i$ is not dense, i.e. for any
fixed $0 < \delta < 1,$ $\vert N_i - n s_n^2 \vert \geq \delta n
s_n^2.$ Using Chernoff bounds on $N_i,$ we have $\prob{W_i = 1} \leq 2
\ \exponent{ - n s_n^2 \delta^2 / 4}.$ The union bound is used to show
that that every cell is dense w.h.p, see
Appendix~\ref{app:eq:denseness} for details.
{ \small{
    \begin{equation}
      \prob{\bigcup_{i=1}^{1/s_n^2} W_i } \leq \frac{1}{s_n^2}
      \prob{W_i} \leq \exponent{- \frac{\theta \delta^2 d_n}{8 \pi}} \to 0.
      \label{eq:denseness}
    \end{equation}
  }
}

\normalsize 

\noindent Now consider the sub-graph formed by nodes in Cell $i;$
denote this subgraph by $G_i.$ We show that
$\prob{\{\cap_{i=1}^{1/s_n^2} \{ G_i \mbox{ is connected} \} \} } \
\to 1.$ This in turn is achieved by showing that for every $i$ there
are no components of size $1,2, \ldots, N_i/2$ in $G_i.$ To simplify
the notation, in the following we will drop the reference to the
parameters $r_n,$ $K_n,$ and $P_n.$

For Cell $i$, define the following events. 

{ \small{
    \begin{eqnarray*}
      S & \subseteq & \{1, 2, \ldots, N_i\} \ \mbox{ is a subset of
        nodes in Cell $i$ }\\
      & & \hspace{.1in} \mbox{with $\vert S \vert \geq 1.$} \\
      C_{i}(S) & := & \mbox{Event that subgraph induced by nodes
        in } S
      \\
      & & \ \mbox{ forms a connected component.} \\
      B_{i}(S) & := & S \mbox{ and } S^c \mbox{ have no
        edges between} \\
      && \hspace{.1in} \mbox{them, where } \ S \cup S^c = \left \{1, 2,
        \ldots, N_i \right \} \\
      A_{i}(S) & := & B_{i}(S) \cap C_{i}(S). \\
      D_i & = & \bigcup_{l=1}^{\lceil N_i/2 \rceil} \bigcup_{S : \vert
        S \vert = l} A_{i}(S).
    \end{eqnarray*}
  } 
}

\normalsize
\noindent Further, let $C_{i, l}$ and $A_{i, l}$ denote, respectively,
$C_{i}\left(S \right)$ and $A_{i}\left(S\right)$ with $\vert S \vert =
l.$ Then the sufficient condition for $G_i$ to be connected w.h.p is
to have $\prob{D_i} \to 0.$ Conditioning on $W_i,$ we have
\begin{eqnarray*}
  \prob{D_i} & = & \sum_{j\in \{0, 1\}} \prob{D_i \vert W_i=j}
  \prob{W_i = j} \\
  & \leq & \prob{D_i \vert W_i=0} + \prob{W_i=1}.
\end{eqnarray*}
The preceding inequality is obtained by using $\prob{W_i = 0} \leq 1$
and $\prob{D_i \vert W_i = 1} \leq 1.$ 

Let $U_{i,l}$ be the random variable that denotes the number of
distinct keys in the component of size $l$ in $G_i.$
Adapting~\cite[(56) from Lemma~10.2]{OYagan12} for each cell, for any
$x \in \{K_n, K_n + 1, \ldots \min(l K_n, P_n) \},$ we
have~\eqref{eq:10.2}.
\begin{eqnarray}
  & & \hspace{-.3in} \prob{A_{i, l}} \ \leq \
  \prob{U_{i,l} \leq x } \exponent{ - \left(\lfloor
      N_i \rfloor - l \right) \frac{K_n^2}{P_n}} \nonumber \\
  & & \hspace{.25in} + \ \prob{C_{l}} \exponent{ -
    \left(\lfloor N_i \rfloor - l \right) \frac{K_n (x+1)}{P_n}
  }. \label{eq:10.2}
\end{eqnarray}
From~\cite[Lemma 10.1 and (69)]{OYagan12}, we know that
\begin{eqnarray*}
  \prob{U_{i,l}  \leq x} & \leq & {P_n \choose x}
  \left( \frac{x}{P_n} \right)^{lK_n} \\
  \prob{C_{i,l} } & \leq & l^{l-2} \beta_n^{l-1}.
\end{eqnarray*}
Now consider all the cells in a tessellation.
{ \small{
    \begin{eqnarray*}
      \prob{\cup_{i=1}^{\left( \frac{1}{s_n} -1 \right)^2} D_i} & \leq
      & \frac{\prob{D_i \vert W_i = 0}}{s_n^2} + \frac{\prob{ W_i
          =1}}{s_n^2}.
    \end{eqnarray*}
  }
  \normalsize 
}
\noindent From~\eqref{eq:denseness}, $(1/s_n^2) \ \prob{W_i = 1} \to
0.$ Thus we focus on showing that $(1/s_n^2) \prob{D_i \vert W_i = 0}
\to 0.$ This implies that all $G_i\left(P_n, K_n, r_n\right)$ are
connected w.h.p.

By using symmetry and union bound, we have
{ \small{
    \begin{eqnarray}
      \frac{ \prob{ D_i \vert W_i=0}}{s_n^2} & = & \left(
        \frac{1}{s_n^2} \right) \ \prob{ \bigcup_{l=1}^{\lceil N_i/2 \rceil}
        \bigcup_{S: \vert S \vert = l} A_{N_i, l}  }
      \nonumber \\
      & & \hspace{-.3in} \leq \left( \frac{1}{s_n^2} \right) \
      \sum_{l=1}^{\lceil N_i/2 \rceil} {N_i \choose l} \prob{A_{N_i,
          l}}. \label{eq:sufficiency:tpt}
    \end{eqnarray}
  } \normalsize 
}

\noindent For the remainder of this section, assume that $n \pi r_n^2
\beta_n =: \alpha \log n.$ The probability of having isolated nodes in
any of the cells is upper bounded as shown below (details are in
Appendix~\ref{app:suff:key:isolated}).
\begin{eqnarray}
  && \hspace{-.2in} \prob{\exists \geq 1 \mbox{ isolated node in any
      of the cells}} \nonumber \\
  && \leq \exponent{-\log n \ \left( \frac{\left( \frac{\alpha \left(1
              - \delta \right) }{2 \pi} - 1 \right)}{2} \right) } \to 0.
  \label{eq:suff:key:isolated}
\end{eqnarray}
Further, the following conditions on the constants are necessary. $0 <
\delta < 1$ and $0 < \mu < 0.44.$ $\lambda, R$ are chosen such that
$\lambda R > \alpha \left(1 - \delta \right) / \left(2 \pi \right).$
We also need $K_n > 2 \log 2 / \mu.$ Further $\sigma, \lambda, \delta,
K_n$ must satisfy
{\small
  {
    \begin{eqnarray*}
      \sigma & \geq & \frac{(1 + \delta) \log 2}{\log \left(
          \frac{e^{\mu}}{\mu^{1 + \mu}} \right)} \\
      1 & > & \max \left \{\frac{e^{2 + \frac{K_n^2}{P_n}} (1 + \delta)}
        {2^{K_n - 2} \sigma}, e^{K_n / P_n} \left(\frac{e^2 (1 +
            \delta)}{\sigma} \right)^{ \lambda} \lambda ^{\left(1 - 2 \lambda
          \right)}\right\}.
    \end{eqnarray*}
  }  
}

\normalsize

Using~\eqref{eq:10.2} in~\eqref{eq:sufficiency:tpt}, we next prove
that all cells in tessellation 1 do not have components of size $2, 3,
\ldots N_i/2.$ Together with~\eqref{eq:suff:key:isolated}, we have
$\prob{\{\cap_{i=1}^{1/s_n^2} \{\mbox{ $G_i$ is connected} \} \}} \to
1.$

Following ~\cite[(61)]{OYagan12} or~\cite{Bollobas01}, the sum term
in~\eqref{eq:sufficiency:tpt} is evaluated in three parts based on the
size of the isolated component $l.$
\begin{enumerate}
\item $2 \leq l \leq R:$ In this case, the number of keys shared by
  the set of nodes which form the isolated component is small and can
  be upper bounded by $(1 + \epsilon) K_n,$ where $0 < \epsilon < 1.$
  $R$ is a small integer, See Appendix~\ref{app:suff:key:case1} for
  details.

  { \small{
      \begin{equation}
        \left( \frac{1}{s_n^2} \right) \ \sum_{i=2}^{R} {N_i \choose
          l} \prob{A_{N_i, l}} \leq \frac{(R-1)
          c_{4}}{n^{0.5 \left( \frac{ (1 - \delta) \alpha}{\pi} -1 \right) }}
        \label{eq:app:suff:key:case1}
      \end{equation}
    }
  }

  \normalsize
  \noindent where $c_{4}$ is an appropriately chosen positive
  constant.

\item $R+1 \leq l \leq L_1(n):$ Here $L_1(n) = \min \left(\lfloor
    N_i/2 \rfloor, \lfloor P_n/ K_n \rfloor - 1 \right).$ In this
  case, the number of keys shared by the set of nodes which form the
  isolated component is linear in the number of nodes $l$ and is upper
  bounded by $\lambda l K_n,$ where $0 < \lambda < 1/2.$ See
  Appendix~\ref{app:suff:key:case2} for details.
  \begin{eqnarray}
    && \hspace{-.3in} \left( \frac{1}{s_n^2} \right) \
    \sum_{i=R+1}^{L_1(n)} {N_i \choose l} \prob{A_{N_i, l}} \leq
    \nonumber \\
    && \hspace{.4in} \frac{c_5}{n^{0.5 \left( \alpha (1 - \delta)/2
          \pi \right)}} + \frac{c_6}{n^{c_7}}.
    \label{eq:app:suff:key:case2}
  \end{eqnarray}
\item $L_1(n) + 1 \leq l \leq N_i/2:$ In this case, the isolated
  component is large, and comparable to the size of the subgraph $G_i$
  in cell $i.$ Thus the number of keys shared by the nodes which form
  the isolated component is upper bounded by $\mu P_n,$ where $0 < \mu
  < 0.44.$ See Appendix~\ref{app:suff:key:case3} for details of the
  following result.
  \begin{eqnarray}
    && \hspace{-.2in} \left( \frac{1}{s_n^2} \right) \
    \sum_{i=L_1(n)+1}^{N_i/2} {N_i \choose l} \prob{A_{N_i,
        l}} \nonumber \\
    && \leq \exponent{- c_8 \ d_n} + \exponent{- c_9 \ d_n}.
    \label{eq:app:suff:key:case3}
  \end{eqnarray}
  Where $c_8 > 0,$ $c_9>((1 - \delta)/4\pi) \left(\frac{\mu K_n}{2} -
    \log 2 \right) .$
  \begin{remark}
    If tighter upper bounds on ${P_n \choose \mu P_n}$ than
    $\left(e/\mu\right)^{\mu P_n}$ are used, then the bound
    in~\eqref{eq:app:suff:key:case3} can be improved in terms of
    larger range of of $\mu;$ i.e. for instance if ${P_n \choose \mu
      P_n} \leq 0.85 \left(e/\mu\right)^{\mu P_n},$ then $0 < \mu \leq
    0.5$ is valid.
  \end{remark}
\end{enumerate}

Combining~\eqref{eq:app:suff:key:case1},~\eqref{eq:app:suff:key:case2}
and~\eqref{eq:app:suff:key:case3}, we have
\begin{eqnarray*}
  \left( \frac{1}{s_n^2} \right) \ \sum_{i=2}^{N_i/2} {N_i \choose
    l} \prob{A_{N_i, l}} & \leq & \frac{c_4 (R+1)}{n^{0.5
      \left( \frac{(1 - \delta) \alpha}{\pi} -1 \right) }} \\
  && \hspace{-2in} + \frac{c_5}{n^{0.5 \left( \frac{\alpha (1 -
          \delta)}{2 \pi} - 1\right) }} + \frac{c_6}{n^{c_7}} + \exponent{- c_8
    \ d_n} + \exponent{- c_9 \ d_n}.
\end{eqnarray*}
Further using appropriate positive constants $c_2, c_3$ and
Lemma~\ref{lemma:equiv:kn2pn}, we have the sufficient condition. Thus
from~\eqref{eq:suff:key:isolated}, \eqref{eq:app:suff:key:case1},
\eqref{eq:app:suff:key:case2} and~\eqref{eq:app:suff:key:case3}, we
have shown that $\prob{T_1} \to 1,$ where $T_i, i = 1$ or $2,$
represents the event that all cells in tessellation $i$ are connected.

\noindent $\prob{T_1 \cap T_2} \to 1$ implies that the entire graph is
connected. We know that $\prob{T_1} \to 1,$ and $\prob{T_2} \to 1.$
Thus
\begin{displaymath}
  \prob{T_1 \cap T_2} = \prob{T_1} + \prob{T_2} - \prob{T_1 \cup T_2},
\end{displaymath}
$\prob{T_1 \cup T_2} \leq 1,$ and $\prob{T_1 \cap T_2} \to 1$ which
completes the proof. $\hfill \square$

\begin{remark}
  Analysis of connectivity of the intersection of the ER and RGG using
  the same technique used in the proof of Theorem~\ref{thm:main} will
  yield an identical result.
\end{remark}

\section{Discussion and Conclusion}
\label{sec:conclusion}

Imposing the random key graph constraint on random geometric graphs
was discussed in \cite{OYagan:secure:pn} where it was conjectured that
the connectivity threshold will be of the form $n \pi r_n^2 \beta_n =
\log n + c_n$ for any $c_n \to \infty.$ We have obtained this up to a
multiplicative constant, as opposed to the additive constant
conjectured in \cite{OYagan:secure:pn}. Further, it may also be
possible to be less restrictive about $n \pi r_n^2$ and $\beta_n.$ As
we mentioned earlier, the minimum degree should be increasing in $n,$
but we believe that can also be made tighter.

\bibliographystyle{IEEEtran}
\bibliography{references}

\begin{thebibliography}{1}
\providecommand{\url}[1]{#1}
\csname url@samestyle\endcsname
\providecommand{\newblock}{\relax}
\providecommand{\bibinfo}[2]{#2}
\providecommand{\BIBentrySTDinterwordspacing}{\spaceskip=0pt\relax}
\providecommand{\BIBentryALTinterwordstretchfactor}{4}
\providecommand{\BIBentryALTinterwordspacing}{\spaceskip=\fontdimen2\font plus
\BIBentryALTinterwordstretchfactor\fontdimen3\font minus
  \fontdimen4\font\relax}
\providecommand{\BIBforeignlanguage}[2]{{%
\expandafter\ifx\csname l@#1\endcsname\relax
\typeout{** WARNING: IEEEtran.bst: No hyphenation pattern has been}%
\typeout{** loaded for the language `#1'. Using the pattern for}%
\typeout{** the default language instead.}%
\else
\language=\csname l@#1\endcsname
\fi
#2}}
\providecommand{\BIBdecl}{\relax}
\BIBdecl

\bibitem{Bollobas01}
B.~Bollobas, \emph{Random Graphs}.\hskip 1em plus 0.5em minus 0.4em\relax
  Cambridge University Press, 2001.

\bibitem{Penrose03}
M.~D. Penrose, \emph{Random Geometric Graphs}.\hskip 1em plus 0.5em minus
  0.4em\relax Oxford University Press, 2003.

\bibitem{OYagan12}
O.~Yagan and A.~M. Makowski, ``Zero-one laws for connectivity in random key
  graphs,'' \emph{IEEE Transactions on Information Theory}, vol.~58, no.~5, pp.
  2983--2999, May 2012.

\bibitem{Dia:conn:rnd:int:graph:Rybarczyk}
K.~Rybarczyk, ``Diameter, connectivity, and phase transition of the uniform
  random intersection graph,'' \emph{Discrete Mathematics}, vol. 311, no.~17,
  pp. 1998 -- 2019, 2011.

\bibitem{Mao:Connectivity:understanding}
G.~Mao and B.~D.~O. Anderson, ``Towards a better understanding of large scale
  network models,'' \emph{IEEE/ACM Transactions on Networking}, vol.~20, no.~2,
  pp. 408--421, 2012.

\bibitem{OYagan:secure:pn}
O.~Yagan, ``{Performance of the Eschenauer-Gligor key distribution scheme under
  an ON-OFF channel},'' \emph{IEEE Transactions on Information Theory},
  vol.~56, no.~6, pp. 3821--3835, Jun. 2012.

\bibitem{GuptaKumar:connectivity}
P.~Gupta and P.~R. Kumar, ``Critical power for asymptotic connectivity in
  wireless networks,'' in \emph{Stochastic Analysis, Control, Optimization and
  Applications: A Volume in Honor of W.H. Fleming}.\hskip 1em plus 0.5em minus
  0.4em\relax Birkhauser, Boston, 1998, pp. 547--566.

\bibitem{Penrose:Longest}
M.~D. Penrose, ``{The longest edge of the random minimal spanning tree},''
  \emph{The Annals of Applied Probability}, vol.~7, no.~2, pp. 340--361, 1997.

\end{thebibliography}

\clearpage
\appendix
\renewcommand{\theequation}{\thesubsection-\arabic{equation}}
\setcounter{equation}{0}

\subsection{Preparatory Definitions and Results}
\label{app:sec:prelims}
\begin{enumerate}
\item Recall the following order notations for comparing functions
  $f(n), g(n)$ as $n \to \infty;$ $f(n) \in \Theta(g(n)), \ f(n) \in
  o(g(n)), \ f(n) \in \omega(g(n)).$ The notations $f(n) =
  \Theta(g(n))$ and $f(n) \in \Theta(g(n))$ are used interchangeably.
  \begin{enumerate}
  \item If $f(n) \in \Theta(g(n)),$ then there are constants $0 < a <
    b < \infty$ and a $N \in \mathbb{N}$ such that $a g(n) \leq f(n)
    \leq b g(n)$ for all $n \geq N.$
  \item If $f(n) \in o\left( g(n) \right),$ then $\lim_{n \to \infty}
    f(n)/g(n) = 0.$
  \item If $f(n) \in \omega\left(g(n)\right),$ then $\lim_{n \to
      \infty} f(n)/g(n) = \infty.$ Clearly $f(n) \in \omega(g(n))
    \Leftrightarrow g(n) \in o(f(n)).$
  \item If $f(n) \sim g(n),$ then $\lim_{n \to \infty} f(n)/g(n) = 1.$
  \end{enumerate}
\item We prove~\eqref{eq:braveinequality} that, for any constant $x,$
  such that $0 < x < 1,$ and any positive integer $n,$
  \begin{equation*}
    \exponent{- \frac{n x}{1 - x} } \ < \ (1 - x)^n \ < \ \exponent{ -
      n x}.
  \end{equation*}
  This implies that $(1 - x)^n \rightarrow 0$ if and only if
  $\exponent{-nx} \rightarrow 0.$
  \begin{proof}
    The upper bound on $(1 - x)^n$ is obvious. The lower bound is
    derived by using the transformation $x = -v/(1-v),$ in the well
    known inequality $e^{v} \geq (1 + v)$ for all $v \in
    \Re.$ 
  \end{proof}
\end{enumerate}

\section{The Details} 

\subsection{Deriving~\eqref{eq:necessary:1}}
\label{app:eq:necessary:1}
\begin{eqnarray*}
  & & \hspace{-.2in} n \prob{Z_1} = n \left(1 - \an \beta_n
  \right)^{n-1} \\
  & = & n \left(1 - (\log n + c_1) / n \right)^{n-1} \\
  & \geq & n \exponent{-(n-1)\left( \frac{(\log n + c_1) / n}{1 -
        (\log n + c_1) / n} \right) } \\
  & = & \exponent{ - (\log n + c_1)\left( \frac{(n-1)}{n - (\log n +
        c_1) } \right) + \log n } \\
  & = & \exponent{ - (\log n + c_1)\left( 1 + \frac{(\log n + c_1
        -1)}{n - (\log n + c_1) } \right) + \log n} \\
  & = & \exponent{-c_1} \ \exponent{ - (\log n + c_1)\left(
      \frac{(\log n + c_1 -1)}{n - (\log n + c_1) } \right)} \\
  & \geq & \exponent{-c_1} \ \exponent{ - (\log n + c_1)\left(
      \frac{(\log n + c_1 )}{n - (\log n + c_1) } \right)} \\
  & = & \exponent{-c_1} \ \exponent{- \frac{(\log n + c_1 )^2}{n -
      (\log n + c_1) } }
\end{eqnarray*}
The third step above uses \eqref{eq:braveinequality}. The inequality
in the penultimate step is obtained by omitting the $(1-)$ in the
numerator of the exponent.

\subsection{Deriving~\eqref{eq:necessary:case2}}
\label{app:eq:necessary:2}

In this case, condition the event $\{Z_1 \cap Z_2\}$ on $S_1 \cap
S_2.$ Thus
{ \small {
    \begin{eqnarray*}
      \prob{Z_1 \cap Z_2} & = & \beta_n \prob{Z_1 \cap Z_2 \vert S_1
        \cap S_2 \neq \emptyset} \\
      && \hspace{-.4in} + (1 - \beta_n) \prob{Z_1 \cap Z_2 \vert S_1
        \cap S_2 = \emptyset}.
    \end{eqnarray*}
  } 
}

\normalsize

\noindent
We know from Lemma~\ref{lemma:equiv:kn2pn} and the assumption that
$K_n^2 / P_n \to 0$ that $\beta_n \to 0.$ Next we calculate the
probability of the joint isolation event $\{Z_1 \cap Z_2\}$ assuming
$S_1 \cap S_2 = \emptyset.$

{ \small{
    \begin{eqnarray*}
      && \hspace{-21pt} 
      \prob{Z_1 \cap Z_2 \vert S_1 \cap S_2 = \emptyset; n_1; n_2;
        n_3; } \\
      && = \ \frac{ {P_n \choose K_n} {P_n -K_n \choose K_n} {P_n -K_n
          \choose K_n}^{n_1 + n_2} {P_n - 2 K_n \choose K_n}^{n_3} }{ {P_n
          \choose K_n}^{2 + n_1 + n_2 + n_3}} \\
      && = (1 - \beta_n) (1 - \beta_n)^{n_1 + n_2} \left(1 -
        \tilde{\beta}_n \right)^{n_3}.
    \end{eqnarray*}
  } 
}

\normalsize

\noindent Where $\beta_n = 1 - \left({P_n - K_n \choose K_n} / {P_n
    \choose K_n } \right)$ and $\tilde{\beta}_n = 1 - \left({P_n - 2
    K_n \choose K_n} / {P_n \choose K_n} \right).$ Let $d = \Vert x_1
- x_2 \Vert.$ Ignoring edge effects, $n_i,$ the number of nodes in
area $B_i,$ is a multinomial with parameters $(n-2, B_i).$ Thus the
conditional joint isolation probability is:
{
  \small{
    \begin{eqnarray}
      &&\hspace{-.2in} \prob{Z_1 \cap Z_2 | d; S_1 \cap S_2 =
        \emptyset } \nonumber \\
      && = \sum_{n_1, n_2, n_3} {n-2 \choose n_1, n_2, n_3} \ B_1^{n_1
        + n_2} B_3^{n_3} \nonumber \\
      && \hspace{0.1in} \times \ \left(1 - 2B_1 - B_3 \right)^{n-2-n_1
        - n_2 -n_3} \ \left( \prob{Z_1 \cap Z_2 | n_1; n_2; n_3} 
      \right) \nonumber \\
      && = \left( 1 - \beta_n \right) \sum_{n_1, n_2, n_3} {n-2
        \choose n_1, n_2, n_3} \ \left(B_3 (1 -
        \tilde{\beta}_n ) \right)^{n_3} \nonumber \\
      && \hspace{0.1in} \times \ \left(1 - 2B_1 - B_3 \right)^{n-2-n_1
        - n_2 -n_3} \left(B_1 ( 1 - \beta_n ) \right)^{n_1 + n_2} \nonumber \\
      && = \left(1 - \beta_n \right) \left(1 - 2 B_1 \beta_n - B_3
        \tilde{\beta}_n \right)^{n-2} \nonumber \\
      && = \left(1 - \beta_n \right) \left(1 - 2 B_1 \beta_n - 2 B_3
        \beta_n + B_3 (2 \beta_n - \tilde{\beta}_n )\right)^{n-2} \nonumber \\
      && = \left(1 - \beta_n \right) \left(1 - 2 \an \beta_n - B_3
        (\tilde{\beta}_n - 2 \beta_n )\right)^{n-2} \nonumber \\
      && = \left(1 - \beta_n \right) \left(1 - 2 \an \beta_n - B_3
        \beta_n \left(\frac{\tilde{\beta}_n}{\beta_n} - 2 \right)
      \right)^{n-2}
      \label{eq:nec:case:2:t0}
    \end{eqnarray}
  } 
}
\normalsize

\noindent
Clearly $\tilde{\beta}_n / \beta_n > 1;$ the following two cases $1 <
\tilde{\beta}_n/ \beta_n < 2$ and $\tilde{\beta}_n > 2 \beta_n$ are
analyzed.
\begin{enumerate}
\item If $1 < \tilde{\beta}_n / \beta_n \leq 2,$ then

  { \small {
      \begin{eqnarray}
        && \hspace{-.2in} \prob{Z_1 \cap Z_2 | d; S_1 \cap S_2 =
          \emptyset} \nonumber \\
        && = \left(1 - \beta_n \right) \left(1 - 2 \an \beta_n - B_3
          \beta_n \left(\frac{\tilde{\beta}_n}{\beta_n} - 2 \right)
        \right)^{n-2} \nonumber \\
        && = \left(1 - \beta_n \right) \left(1 - 2 \an \beta_n + B_3
          \beta_n \left \Vert \frac{\tilde{\beta}_n}{\beta_n} - 2 \right \Vert
        \right)^{n-2} \nonumber \\
        && \leq \left(1 - \beta_n \right) \left(1 - 2 \an \beta_n +
          \an \beta_n \left \Vert \frac{\tilde{\beta}_n}{\beta_n} - 2 \right
          \Vert \right)^{n-2} \nonumber \\
        && = \left(1 - \beta_n \right) \left(1 - \an \beta_n \left(
            2 - \left \Vert \frac{\tilde{\beta}_n}{\beta_n} - 2 \right \Vert
          \right)\right)^{n-2}
        \label{eq:nec:case:2:t1}
      \end{eqnarray}
    }
  }

  \normalsize
  \noindent
  The inequality uses $B_3 < \an$ and the other steps are algebraic
  manipulations.

\item If $\tilde{\beta}_n / \beta_n > 2,$ then 
  { \small {
      \begin{eqnarray}
        && \hspace{-.2in} \prob{Z_1 \cap Z_2 | d; S_1 \cap S_2 =
          \emptyset } \nonumber \\
        && = \left(1 - \beta_n \right) \left(1 - 2 \an \beta_n - B_3
          \beta_n \left(\frac{\tilde{\beta}_n}{\beta_n} - 2 \right)
        \right)^{n-2} \nonumber \\
        && \leq \left(1 - \beta_n \right) \left(1 - 2 \an
          \beta_n\right)^{n-2}. \label{eq:nec:case:2:t2}
      \end{eqnarray}
    }\normalsize
  }
  The last expression is obtained by neglecting the positive term in
  the preceding step.
\end{enumerate}
Clearly from~\eqref{eq:nec:case:2:t1} and~\eqref{eq:nec:case:2:t2}, we
need the bound in~\eqref{eq:nec:case:2:t1}. Further
using~\eqref{eq:braveinequality}, and $1 - \beta_n \leq 1,$ we have
{ \small {
    \begin{eqnarray*}
      && \hspace{-.3in} \prob{Z_1 \cap Z_2 | d; S_1 \cap S_2 =
        \emptyset} \\
      && \leq \exponent{- (n-2) \left( 2 - \left \Vert
            \frac{\tilde{\beta}_n}{\beta_n} - 2 \right \Vert \right) \an \beta_n
      }.
    \end{eqnarray*}
  }\normalsize
}
We find an upper bound to the probability of joint isolation when $S_1
\cap S_2 \neq \emptyset$ as follows. Recall that $B_1, B_2, B_3$ are
regions in the circles, see Fig.~\ref{fig:circles}, and $x_i$ is the
location of node $i.$ Now, for each node $i \neq 1 \mbox{ or } 2,$ we
need it to not be connected to either node 1 or node 2. 
\ignore{This requires the intersection of the following events
  \begin{enumerate}
  \item[$A_1$] $i \notin B_1$ and shares at least one key with node $1.$
  \item[$A_2$] $i \notin B_2$ and shares at least one key with node $2.$
  \item[$A_3$] $i \notin B_3$ and shares at least one key with either node $1$
    or node $2.$
  \end{enumerate}
  For a given node $i,$ the complementary events $A_1^{c}, A_2^{c}$
  and $A_3^{c}$ are disjoint, so the probability of the union of the
  complements is the sum of their probabilities.

  Consider the event that a node $i, \ i > 2,$ is not connected to
  either node $1$ or node $2$ but it is with in $r_n$ of either of the
  two nodes. Clearly, the event that any two other nodes $i$ and $j,$
  where $i \neq j,$ are dependent, because the event ($A_3$) gives
  information about the size of the combined key ring of nodes 1 and
  2. 
}
Observe that the probability that---(1) nodes $1$ and $2$ are jointly
isolated from $n_3$ nodes in $B_3$ and (2) nodes $1$ and $2$ share $x$
keys, is
{\small
  {
    \begin{displaymath}
      P_x := \frac{{P_n \choose x} {P_n - x \choose K_n - x} {P_n - K_n
          \choose K_n - x} \left( {P_n -2 K_n + x \choose K_n} \right)^{n_3} }{
        \left({P_n \choose K_n} \right)^{2 + n_3}}.
    \end{displaymath}
  }\normalsize
}
Observe that $P_x / P_{x+1},$ for $1 \leq x \leq K_n,$ is
{\small
  {
    \begin{eqnarray*}
      \frac{P_x}{P_{x+1}} & = & \frac{{P_n \choose x} {P_n -x \choose
          K_n - x} {P_n -K_n \choose K_n - x} \left({P_n - 2K_n + x \choose K_n}
        \right)^{n_3}}{{P_n \choose x+1} {P_n -x-1 \choose K_n - x-1} {P_n
          -K_n \choose K_n - x-1} \left({P_n - 2K_n + x+1 \choose K_n}
        \right)^{n_3}} \\
      & = & \frac{ \frac{P!}{x! \left(K_n - x \right)!^2 \left(P_n -
            2K_n + x \right)!}  \left(\frac{\left( P_n - 2K_n + x \right)!}{\left(
              P_n - 3K_n + x \right)!}  \right)^{n_3}}{ \frac{P!}{(x+1)! \left(K_n -
            x - 1 \right)!^2 \left(P_n - 2K_n + x + 1 \right)!}
        \left(\frac{\left( P_n - 2K_n + x + 1\right)!}{\left( P_n - 3K_n + x
              +1\right)!}  \right)^{n_3}} \\
      & = & \frac{ (x+1) \left(P_n - 2K_n + x + 1\right)}{\left(K_n -x
        \right)^2} \left(\frac{P_n - 3K_n + x + 1}{P_n - 2K_n + x + 1} 
      \right)^{n_3}\\
      & = & \frac{ (x+1) \left(P_n - 2K_n + x + 1\right)}{\left(K_n -x
        \right)^2} \left(\frac{1 - \frac{3K_n - x - 1}{P_n} }{1 - \frac{2K_n -
            x - 1}{P_n}} \right)^{n_3}
    \end{eqnarray*}
  }\normalsize
}
See that the second term $\to 1$ since $K_n^2/P_n \to 0.$ The first
term is large and is positive since it is a reciprocal of $K_n^2/P_n.$
Thus probability of nodes being jointly isolated where $S_1 \cap S_2 =
\emptyset$ is an upper bound to that of the probability when $S_1 \cap
S_2 \neq \emptyset.$ Thus we have
{\small
  {
    \begin{displaymath}
      \prob{Z_1 \cap Z_2 | d} \ \leq \exponent{- (n-2) \left( 2 -
          \left \Vert \frac{\tilde{\beta}_n}{\beta_n} - 2 \right \Vert \right)
        \an \beta_n }.  \ \square
    \end{displaymath}
  }\normalsize
}
\subsection{Deriving~\eqref{eq:necessary:case3}}
\label{app:eq:necessary:3}

Conditioned on $n_1,$ $n_2,$ and $n_3,$ $\prob{Z_1 \cap Z_2 | n_1,
  n_2, n_3}$ is identical to the previous case.
{ \small{
    \begin{displaymath}
      \prob{Z_1 \cap Z_2 | n_1, n_2, n_3} = \left(1 - \beta_n \right)
      \left(1 - \beta_n \right)^{n_1 + n_2} \left(1 - \tilde{\beta}_n
      \right)^{n_3}.
    \end{displaymath}
  }  \normalsize
}
As before $n_1,$ $n_2$ and $n_3$ depend on, respectively, $B_1,$ $B_2$
and $B_3,$ which in turn depends on $d.$ Recall that $\tilde{\beta}_n
:= 1 - \left({P_n - 2 K_n \choose K_n} / {P_n \choose K_n} \right).$
The two nodes 1 and 2 should not share a key and also should be
isolated from all their neigbhours, identical
to~\eqref{eq:nec:case:2:t0}.
{ \small{
    \begin{eqnarray*}
      \prob{Z_1 \cap Z_2 \vert d} = \left(1 - \beta_n \right) \left(1
        - 2 \an \beta_n - B_3 \beta_n \left(\frac{\tilde{\beta}_n}{\beta_n} -
          2 \right) \right)^{n-2}
    \end{eqnarray*}
  }\normalsize
}
Identical to~\eqref{eq:nec:case:2:t1} and~\eqref{eq:nec:case:2:t2}, we
have
  { \small {
      \begin{eqnarray*}
        \prob{Z_1 \cap Z_2 | d} \ \leq \ \exponent{- (n-2) \left( 2 -
            \left \Vert \frac{\tilde{\beta}_n}{\beta_n} - 2 \right \Vert \right)
          \an \beta_n }.
      \end{eqnarray*}
    }\normalsize
  }

\subsection{Deriving~the asymptotic relation between $\tilde{\beta}_n$
  and $\beta_n$}
\label{app:beta:tilde}
The main result of this subsection is
Lemma~\ref{lemma:beta:tilde}. Lemma~\ref{lemma:fn:gn:limit} is
necessary to prove Lemma~\ref{lemma:beta:tilde}.
\begin{lemma}
  Let $f_n, g_n$ be distinct sequences in $n$ such that $f_n \to 0,$
  $g_n \to 0$ and $f_n/ g_n \to 1$ as $n \to \infty.$ Then
  \begin{displaymath}
    \frac{1 - \exponent{-f_n}}{\exponent{g_n} - 1} \to 1.
  \end{displaymath}
  \label{lemma:fn:gn:limit}
\end{lemma}
\begin{proof}
  By using the standard binomial expansions of the exponential
  function, we have
  { \small{
      \begin{eqnarray*}
        \frac{1 - \exponent{-f_n} }{\exponent{g_n} - 1} & = &
        \frac{f_n - \frac{f_n^2}{2!} + \frac{f_n^3}{3!} \ldots }{ g_n +
          \frac{g_n^2}{2!} + \frac{g_n^3}{3!} \ldots } \\
        & = & \left( \frac{f_n}{g_n} \right) \left( \frac{1 -
            \frac{f_n}{2!} + \frac{f_n^2}{3!} \ldots }{ 1 + \frac{g_n}{2!} +
            \frac{g_n^2}{3!}  \ldots }\right) \\
        & \to & \left( 1 \right) \ \left( \frac{1 - o(1)}{1 + o(1)}
        \right) \to 1.
      \end{eqnarray*}
    }\normalsize
  }
\end{proof}

\noindent The bounds on ratios of binomials
using~\eqref{eq:braveinequality} are:
\begin{eqnarray*}
  \frac{{P_n - K_n \choose K_n}}{{P_n \choose K_n}} & = &
  \frac{ \left( P_n - K_n \right)!^2}{ \left( P_n - 2K_n \right)!\left(
      P_n \right)!} \\
  & = & \prod_{i=1}^{K_n} \left( 1 - \frac{K_n }{P_n - K_n + i}
  \right), 
\end{eqnarray*}
{ 
  \small{
    \begin{eqnarray} 
      && \hspace{-.4in} \left(1 - \frac{K_n}{P_n - K_n +1}
      \right)^{K_n} \leq \frac{{P_n - K_n \choose K_n}}{{P_n\choose
          K_n}} \leq \left(1 - \frac{K_n}{P_n} \right)^{K_n},
      \nonumber \\
      && \hspace{-.4in} \exponent{- \frac{K_n^2}{P_n - 2K_n + 1}}
      \leq \frac{{P_n - K_n \choose K_n}}{{P_n \choose K_n}} \leq
      \exponent{- \frac{K_n^2}{P_n}}. \label{ineq:p-k-choose-k}
    \end{eqnarray}
  } \normalsize
}

\begin{eqnarray*}
  \frac{{P_n - 2 K_n \choose K_n}}{{P_n - K_n \choose K_n}} & = &
  \frac{ \left( P_n - 2K_n \right)!^2}{ \left( P_n - 3K_n \right)!\left(
      P_n - K_n \right)!} \\
  & = & \prod_{i=1}^{K_n} \left( 1 - \frac{K_n }{P_n - 2 K_n + i}
  \right) 
\end{eqnarray*}
{ 
  \small{
    \begin{eqnarray}
      && \hspace{-.4in} \left(1 - \frac{K_n}{P_n - 2K_n +1}
      \right)^{K_n} \leq \frac{{P_n - 2 K_n \choose K_n}}{{P_n - K_n
          \choose K_n}} \leq \left(1 - \frac{K_n}{P_n - K_n}
      \right)^{K_n} \nonumber \\ 
      && \hspace{-.4in} \exponent{- \frac{K_n^2}{P_n - 3K_n + 1}}
      \leq \frac{{P_n - 2 K_n \choose K_n}}{{P_n - K_n \choose K_n}}
      \leq \exponent{- \frac{K_n^2}{P_n - K_n}}. 
      \label{ineq:p-2k-choose-k}
    \end{eqnarray}
  } \normalsize
}
\begin{lemma}
  If $K_n^2/ P_n \to 0,$ then for any $0 < \epsilon < 1,$ the
  following holds for all $n$ sufficiently large.
  \begin{displaymath}
    1 - \epsilon \leq \frac{\tilde{\beta}_n}{\beta_n} - 1 \leq 1 +
    \epsilon.
  \end{displaymath}
  \label{lemma:beta:tilde}
\end{lemma}
\begin{proof}
  We rewrite $\tilde{\beta}_n / \beta_n - 1 $ and derive upper and
  lower bounds using \eqref{ineq:p-k-choose-k} and
  \eqref{ineq:p-2k-choose-k}. 
  {\small{
      \begin{eqnarray*}
        \frac{\tilde{\beta}_n}{\beta_n} - 1 & = & \frac{1 - \frac{ {P_n -
              2K_n \choose K_n} }{ {P_n \choose K_n} }}{1 - \frac{ {P_n - K_n
              \choose K_n}}{ {P_n \choose K_n}} } -1 \\
        & & \hspace{-.2in} = \frac{ {P_n \choose K_n} - {P_n - 2 K_n
            \choose K_n} }{{P_n \choose K_n} - {P_n - K_n \choose K_n} } -1 \ = \
        \frac{ {P_n - K_n \choose K_n} - {P_n - 2 K_n \choose K_n} }{{P_n
            \choose K_n} - {P_n - K_n \choose K_n} } \\
        & = & \left( \frac{ {P_n - K_n \choose K_n}}{{P_n \choose K_n}}
        \right) \left( \frac{ 1 - \frac{{P_n - 2 K_n \choose K_n}}{{P_n - K_n
                \choose K_n}} }{1 - \frac{{P_n - K_n \choose K_n} }{{P_n \choose
                K_n}}} \right). \\
        & \geq & \exponent{- \frac{K_n^2}{P_n - 2 K_n + 1} } \
        \left(\frac{1 - \exponent{- \frac{K_n^2}{P_n - K_n} }}{1 - \exponent{-
              \frac{K_n^2}{P_n - 2 K_n + 1} } } \right) \\
        & = & \frac{1 - \exponent{- \frac{K_n^2}{P_n - K_n}
          }}{\exponent{ \frac{K_n^2}{P_n - 2 K_n + 1} } - 1 }.
      \end{eqnarray*}
    } \normalsize
  }

  \noindent The first inequality uses the
  bounds~\eqref{ineq:p-k-choose-k}
  and~\eqref{ineq:p-2k-choose-k}. Further, using
  Lemma~\ref{lemma:fn:gn:limit} in the final expression, for $0 <
  \epsilon < 1,$
  \begin{displaymath}
    \frac{\tilde{\beta}_n}{\beta_n} - 1 \geq 1 - \epsilon.
  \end{displaymath}
  The upper bound is derived along the same lines and is as under.
  {\small{
      \begin{eqnarray*}
        && \hspace{-.5in} \frac{\tilde{\beta}_n}{\beta_n} - 1 \leq
        \exponent{- \frac{K_n^2}{P_n} } \ \left(\frac{1 -
            \exponent{- \frac{K_n^2}{P_n - 3 K_n + 1} }}{1 - \exponent{-
              \frac{K_n^2}{P_n} } } \right) \\
        && \hspace{-.1in}= \frac{1 - \exponent{- \frac{K_n^2}{P_n -
              3K_n + 1} }}{\exponent{ \frac{K_n^2}{P_n } } - 1 } \leq 1 + \epsilon.
      \end{eqnarray*}
    } \normalsize
  }

  \noindent The final expression is true for any $0 < \epsilon < 1$ by
  using Lemma~\ref{lemma:fn:gn:limit}.
\end{proof}
Thus $\tilde{\beta}_n / \beta_n \to 2.$

\subsection{Deriving~\eqref{eq:prob:z12}}
\label{app:eq:prob:z12}

Recall that $\an \beta_n = (\log n + c_1) / n,$ with $0 < c_1 <
\infty$ and $n\an = d_n$ where $d_n \in o(n).$
\begin{enumerate}
\item The upper bound on joint probability that Nodes 1 and 2 are
  isolated if $d \geq 2 r_n,$ is obtained
  using~\eqref{eq:necessary:case1} as follows.
  {\small {
      \begin{eqnarray}
        && \hspace{-.4in} \left( 1 - 4 \an \right) {n \choose 2} \left(1 -
          2\an \beta_n \right)^{n-2} \nonumber \\
        && \hspace{-.4in} \leq {n \choose 2} \left(1 - 2 \an \beta_n
        \right)^{n-2} \leq \frac{1}{2} \exponent{- \left(n-2 \right) 2 \an
          \beta_n + 2 \log n } \nonumber \\
        && \hspace{-.4in} = \frac{1}{2} \exponent{- 2 \left( \log n +
            c_1 \right) + 4 \frac{\log n + c_1 }{n} + 2 \log n } \nonumber \\
        && \hspace{-.4in} = \exponent{-c_1 } \ \frac{\exponent{-c_1 +
            \frac{4 \left(\log n + c_1\right)}{n}
          }}{2} \label{eq:necessary:case1:final}
      \end{eqnarray}
    }\normalsize }

  The first step uses $1 - 4 \an \leq 1;$ the second step
  uses~\eqref{eq:braveinequality} and $n-1 \leq n;$ the rest are
  algebraic manipulations.
\item The upper bound on joint probability that Nodes 1 and 2 are
  isolated if $r_n \leq d \leq 2 r_n,$ and $0 < d \leq r_n$ are
  obtained together using~\eqref{eq:necessary:case2}
  and~\eqref{eq:necessary:case3}. ${n \choose 2} \an \prob{Z_1 \cap
    Z_2 \vert d < r_n} + {n \choose 2} 3 \an \prob{Z_1 \cap Z_2 \vert
    r_n < d < 2r_n} $ is upper bounded as follows:
  { \small{
      \begin{eqnarray}
        && \hspace{-.4in} \frac{{n \choose 2} 4 \an}{n^2} \exp \left(
          \log n \left[\left\Vert \frac{\tilde{\beta}_n}{\beta_n} - 2
            \right\Vert - \frac{c_1 \left(2 - \Vert
                \frac{\tilde{\beta}_n}{\beta_n} - 2 \Vert \right)}{ \log n}
          \right. \right. \nonumber \\
        && \hspace{.3in} \left. \left. + \frac{\left(4 - 2 \left(\Vert
                  \frac{\tilde{\beta}_n}{\beta_n} - 2 \Vert \right) \right) \an
              \beta_n}{\log n } \right] \right) \nonumber \\
        && \hspace{-.4in} \leq 2 \exp \left(\log \an + \log n
          \left[\left\Vert \frac{\tilde{\beta}_n}{\beta_n} - 2 \right\Vert -
            \frac{c_1 \left(2 - \Vert \frac{\tilde{\beta}_n}{\beta_n} - 2 \Vert
              \right)}{ \log n} \right. \right. \nonumber \\
        && \hspace{.2in} \left. \left. + \frac{\left(4 - 2 \left(\Vert
                  \frac{\tilde{\beta}_n}{\beta_n} - 2 \Vert \right) \right) \an
              \beta_n}{\log n } \right] \right) \nonumber \\
        && \hspace{-.4in} = 2 \exp \left(- \log n \left[1 - \left\Vert
              \frac{\tilde{\beta}_n}{\beta_n} - 2 \right\Vert - \frac{\log \left(
                d_n \right)}{\log n} \right. \right. \nonumber \\
        && \hspace{-.1in} \left. \left.  + \frac{c_1 \left(2 - \Vert
                \frac{\tilde{\beta}_n}{\beta_n} - 2 \Vert \right)}{ \log n} -
            \frac{\left(4 - 2 \left(\Vert \frac{\tilde{\beta}_n}{\beta_n} - 2
                  \Vert \right) \right) \an \beta_n}{\log n } \right]\right) \nonumber
        \\
        && \hspace{-.4in} \leq \frac{2}{n^{\epsilon}}. \nonumber
      \end{eqnarray}
    }
    \normalsize }

  \noindent In the penultimate step, clearly $\log \left(d_n \right)/
  \log n <1,$ since $d_n \in o(n).$ Let $1 - \log (d_n)/\log n > 2
  \epsilon,$ for some $\epsilon > 0,$ then the result follows
  directly. Thus for large $n,$ using Lemma~\ref{lemma:beta:tilde} the
  final bound can easily be derived
\end{enumerate}

\subsection{Deriving~\eqref{eq:denseness}}
\label{app:eq:denseness}

Recall that $n \pi r_n^2 = \log n + d_n,$ for $d_n \in o(n)$ and $d_n
\in \omega(\log n),$ and $s_n^2 = \theta r_n^2,$ for $0 < \theta < 1.$
Using Chernoff bounds on $N_i,$ we have
{ \small{
    \begin{eqnarray*} 
      \prob{N_i \leq (1 - \delta) n s_n^2} & \leq & \exponent{ -
        \frac{n s_n^2 \delta^2}{2} },\\
      \prob{N_i \geq (1 + \delta) n s_n^2} & \leq & \exponent{ -
        \frac{n s_n^2 \delta^2}{4} }, \\
      \prob{W_i = 1} & \leq & 2 \ \exponent{ - \frac{n s_n^2
          \delta^2}{4} }.
    \end{eqnarray*}
  }\normalsize
}
The following union bound argument proves that that every cell is
dense w.h.p.
{ \small{
    \begin{eqnarray*}
      && \hspace{-.2in} \prob{\bigcup_{i=1}^{1/s_n^2} W_i } \leq
      \frac{1}{s_n^2} \prob{W_i} \\
      &&  \hspace{-.2in} = \exponent{ - \frac{\theta \delta^2}{4 \pi} \left( \log n +
          d_n \right) + \log \left(\frac{2 n \pi}{\theta \left( \log n + d_n
            \right)} \right) } \\
      && \hspace{-.2in} = \exponent{ -d_n \left( \frac{\theta
            \delta^2}{4 \pi} \left(1 + \frac{\log n}{d_n} \right) - \frac{\log
            n}{d_n} - \frac{\log \left( \frac{2 \pi}{\theta \left(\log n + d_n
                \right) } \right) }{d_n} \right) }.
    \end{eqnarray*}
  }\normalsize
}
In the final expression, the second and third term inside parenthesis
(multiplying $d_n$) are negligible ($\to 0$) since $d_n \in
\omega(\log n).$ Hence $ \prob{\bigcup_{i=1}^{1/s_n^2} W_i } \leq
\exponent{-d_n \left( \theta \delta^2/ 8 \pi\right)}. \hfill \square$



\subsection{Deriving~\eqref{eq:suff:key:isolated}}
\label{app:suff:key:isolated}

Consider Cell $1$ in tessellation $1.$ Let $G_1$ be the sub-graph
formed by the set of $N_1$ nodes in Cell $1.$ Let $Z_i$ be the
indicator variable such that if $Z_i=1$ then in $G_1$ node $i$ is
isolated. Recall that $N_1/ \left(n s_n^2 \right) \in (1- \delta, 1 +
\delta).$ We have the following:
\begin{eqnarray*}
  \prob{Z_i = 1} & = & \left(1 - \left( 1 - \frac{{P_n - K_n \choose
          K_n}}{{P_n \choose K_n}} \right) \right)^{N_1} \\
  & \leq & \exponent{- N_1 \beta_n} \ \leq \ \exponent{- (1 - \delta)
    n s_n^2 \beta_n} \\
  & = & \exponent{- \frac{\alpha \left(1 - \delta \right)}{2 \pi} \log
    n}.
\end{eqnarray*}
The above uses~\eqref{eq:braveinequality} and $ n \pi r_n^2 \beta_n =
\alpha \log n.$ Let $Z(G_i)$ be the indicator variable that there are
no isolated nodes in the subgraph $G_i,$ i.e. in the sub-graph formed
by nodes in cell $i.$ We know from the preceding result that
{\small
  {
    \begin{eqnarray*}
      \prob{Z(G_i) = 1} & \leq & N_i \prob{Z_i = 1} \\
      & \leq & \exponent{\log \left( N_i \right) - \frac{\alpha
          \left(1 - \delta \right)}{2 \pi} \log n} \\
      & & \hspace{-.5in} \leq \exponent{\log \left( 1+ \delta \right)
        + \log \left( n s_n^2 \right) - \frac{\alpha \left(1 - \delta
          \right)}{2 \pi} \log n}.
    \end{eqnarray*}
  }\normalsize
}

\noindent Finally the probability that there are no isolated nodes in
any of the cells is bounded above as follows.
{\small
  {
    \begin{eqnarray*}
      \prob{\bigcup_{i=1}^{\frac{1}{s_n^2}} Z(G_i) } & \leq &
      \frac{1}{s_n^2} \prob{Z(G_i)} \\
      & & \hspace{-1in} \leq \exponent{\log \left( 1+ \delta \right)
        + \log \left( n s_n^2 \right) - \frac{\alpha \left(1 - \delta
          \right)}{2 \pi} \log n - \log s_n^2}. \\
      & & \hspace{-1in} \leq \exponent{\log \left( 1+ \delta \right)
        + \log n - \frac{\alpha \left(1 - \delta \right)}{2 \pi} \log n } \\
      & & \hspace{-1in} \leq \exponent{-\log n \left( \frac{\alpha
            \left(1 - \delta \right) }{2 \pi} - 1 - \frac{\log \left(1 +
              \delta \right)}{\log n} \right)} \\
      & & \hspace{-1in} \leq \exponent{-\log n \ \left( \frac{\left(
              \frac{\alpha \left(1 - \delta \right) }{2 \pi} - 1 \right)}{2} \right)
      } \to 0.
    \end{eqnarray*}
  }\normalsize
}

\noindent From our initial assumption on $\alpha,$ i.e. $\alpha > 2
\pi / (1 - \delta),$ we can see the final step.

\subsection{Deriving~\eqref{eq:app:suff:key:case1}}
\label{app:suff:key:case1}

In this section, we derive the probability of having isolated
components with small finite cardinality. We analyze the first term
of~\eqref{eq:10.2} with the number of keys shared by the isolated
component, $x = \left( 1 + \epsilon \right) K_n,$ as
in~\cite{OYagan12} with $0 < \epsilon < 1.$ Recall the following $n
s_n^2 = d_n/ (2 \pi)$ and using Lemma~\ref{lemma:equiv:kn2pn}, $n \pi
r_n^2 \left( K_n^2 / P_n \right) = \alpha \log n$ where $\alpha > 2
\pi / (1 - \delta)$ and $d_n \in \omega(\log n), \ d_n \in o(n);$
$K_n^2 / P_n \to 0.$
{ \small{
    \begin{eqnarray}
      & & \hspace{-.2in} \frac{1}{s_n^2} {N_i \choose l}
      \prob{U_{l} \leq x} \exponent{-\left(N_i - l
        \right) \frac{K_n^2}{P_n} } \nonumber \\
      & & \leq \frac{N_i^l }{s_n^2} {P_n \choose x} \left(
        \frac{x}{P_n} \right)^{l K_n} \exponent{ - \left(N_i - l\right)
        \frac{K_n^2}{P_n}} \nonumber\\
      & & \leq \frac{\left((1 + \delta) n s_n^2\right)^l }{s_n^2}
      {P_n \choose x} \left( \frac{x}{P_n} \right)^{l K_n} \exponent{ -
        \left(N_i - l\right) \frac{K_n^2}{P_n}} \nonumber\\
      & & \hspace{-.2in} \leq \frac{\left((1 + \delta) n
          s_n^2\right)^l }{s_n^2} \left( \frac{e P_n}{x} \right)^{x} \left(
        \frac{x}{P_n} \right)^{l K_n} \exponent{ - \left(N_i - l\right)
        \frac{K_n^2}{P_n}}. \hspace{.2in}\label{eq:suff:1:temp:1}
    \end{eqnarray}
  }\normalsize
}
The first and third inequalities use the bounds on the binomial while
the second one uses the bound on $N_i.$ Next we use $x = \lfloor (1 +
\epsilon) K_n \rfloor.$ Let $\Gamma(\epsilon) := e^{\frac{1 +
    \epsilon}{1 - \epsilon}} (1 + \epsilon).$ 
{ \small{
    \begin{eqnarray}
      && \hspace{-.2in} \eqref{eq:suff:1:temp:1} \nonumber \\
      && = \frac{\left((1 + \delta) n s_n^2\right)^l }{s_n^2} \left(
        \frac{e P_n}{\lfloor (1 +\epsilon) K_n \rfloor} \right)^{\lfloor (1
        +\epsilon) K_n \rfloor} \left( \frac{\lfloor (1 +\epsilon) K_n
          \rfloor}{P_n} \right)^{l K_n} \nonumber \\
      && \hspace{.8in} \times \exponent{ - \left(N_i - l\right)
        \frac{K_n^2}{P_n}} \nonumber \\
      && \hspace{-.2in} \leq \frac{\left((1 + \delta) n s_n^2\right)^l
      }{s_n^2} \left( e^{\frac{1 + \epsilon}{l - 1 - \epsilon} }
        \frac{\lfloor (1 +\epsilon) K_n \rfloor}{P_n} \right)^{K_n \left(l - 1
          - \epsilon \right) } \hspace{-.2in} \exponent{ - \left(N_i - l\right)
        \frac{K_n^2}{P_n}} \nonumber \\
      & & \leq \frac{\left((1 + \delta) n s_n^2\right)^l }{s_n^2}
      \left(\Gamma(\epsilon) \frac{K_n^2}{P_n} \right)^{K_n(l - 1 -
        \epsilon)} \exponent{-\left(N_i - l \right)
        \frac{K_n^2}{P_n}} \label{eq:suff:1:temp:2}
    \end{eqnarray}
  }\normalsize
}
The first step is a substitution, the first inequality uses $lK_n -
\lfloor K_n(1 + \epsilon) \rfloor \leq K_n (l - 1 - \epsilon).$ Now
since $K_n^2/ P_n \to 0,$ $\Gamma(\epsilon) K_n^2/ P_n < 1,$ we have
{ \small{
    \begin{eqnarray}
      && \hspace{-.2in} \eqref{eq:suff:1:temp:2} \leq \frac{ \left(
          (1 + \delta) n s_n^2\right)^l }{s_n^2} \left(\Gamma(\epsilon)
        \frac{K_n^2}{P_n} \right)^{2(l - 1 - \epsilon)} \exponent{-\left(N_i -
          l \right) \frac{K_n^2}{P_n} } \nonumber \\
      && \leq \left((1 + \delta)^l \Gamma(\epsilon)^{2(l-1-\epsilon)}
      \right) \ \frac{\left(n s_n^2\right)^l }{s_n^2} \left( \frac{\alpha
          \log n}{d_n} \right)^{2(l - 1 - \epsilon)} \nonumber \\
      && \hspace{.3in} \times \exponent{ - \left( 1 - \delta -
          \frac{l}{ns_n^2} \right) \left( \frac{\alpha}{2 \pi} \right) \log n }
      \nonumber
    \end{eqnarray}
  }\normalsize
}
Let $c_{11} = (1 + \delta) \Gamma(\epsilon)^{2 - \frac{1+\epsilon}{l}},$
then the preceding expression is
{\small
  {
    \begin{eqnarray}
      && \hspace{-.2in} = c_{11}^l\exp \left( - \left( 1 - \delta -
          \frac{l}{ns_n^2} \right) \left( \frac{\alpha}{2 \pi} \right) \log n +
        l \log \left(\frac{d_n}{2 \pi} \right) \right. \nonumber \\
      && \left. - \log d_n + \log (2 \pi n ) + 2 (l - 1 - \epsilon)
        \left(\log \left( \frac{\alpha \log n}{d_n} \right) \right) \right) \nonumber \\
      && \hspace{-.2in} \leq c_{11}^l\exp \left( - \log n \left(
          \frac{\alpha}{2 \pi} \left( 1 - \delta - \frac{1 + R}{ns_n^2} \right)
          - 1 + \right. \right. \nonumber \\
      && \left. \left. \frac{(l + 1 - 2 \epsilon) \log d_n}{ \log n }
          - \frac{(l-1) \log 2 \pi }{\log n} + \frac{2 (l - 1- \epsilon) \log
            (\alpha \log n)}{\log n} \right) \right) \nonumber \\
      && \hspace{-.2in} \leq c_{11}^l \exponent{- \log n \ \left(
          \frac{ \left( \frac{\alpha (1 - \delta) }{2\pi} \right) - 1 }{2}
        \right) }. \label{eq:suff:app:1:1}
    \end{eqnarray}
  }\normalsize
}
The final expression $\to 0$ since $\alpha > 2 \pi / (1 - \delta).$

Next consider the second term in the RHS of~\eqref{eq:10.2}.
{ \small{
    \begin{eqnarray}
      && \frac{1}{s_n^2} {N_i \choose l} \prob{C_{i, l}} \exp
      \left(-\left(N_i - l \right) \frac{K_n \left(1 + K_n (1 + \epsilon)
          \right)}{P_n} \right) \nonumber \\
      && \leq \frac{\left( e N_i \right)^l }{l^l s_n^2} l^{l-2}
      \beta_n^{l-1} \exp \left(-\left(N_i - l \right) \frac{K_n^2 \left(1 +
            \epsilon \right)}{P_n} \right) \nonumber \\
      && \leq \frac{\left((1 + \delta) e \right)^l }{l^2 s_n^2}
      \left( n s_n^2 \right)^{l} \beta_n^{l-1} \exp \left(-\left(N_i - l
        \right) \frac{K_n^2 \left(1 + \epsilon \right)}{P_n} \right) \nonumber
      \\
      && \leq \frac{\left((1 + \delta) e \right)^l }{l^2 } n \left(
        n s_n^2 \beta_n \right)^{l-1} \nonumber \\
      && \hspace{.2in} \times \exponent{-\left[ 1 - \delta -
          \frac{l}{ns_n^2} \right] (1 + \epsilon) n s_n^2 \frac{K_n^2}{P_n}}
      \nonumber \\
      && = \frac{\left((1 + \delta) e \right)^l }{l^2 } \exp
      \left(\log n + (l-1) \log \left( \frac{\alpha \log n}{2 \pi} \right)
      \right.\nonumber \\
      && \hspace{.2in} \left.-\left[ 1 - \delta - \frac{l}{ns_n^2}
        \right] (1 + \epsilon) \frac{\alpha}{2 \pi} \log n \right)
      \nonumber \\
      && = \frac{\left((1 + \delta) e \right)^l }{l^2 } \exp \left( -
        \log n \left( \frac{(1 + \epsilon)(1 - \delta)\alpha}{2\pi} - 1 -
          \frac{\alpha l (1 + \epsilon)}{d_n} \right. \right. \nonumber \\
      && \hspace{.2in} \left. \left. - \frac{(l-1) \log(\alpha/
            2\pi)}{\log n} - \frac{(l-1) \log \log n}{\log n} \right)
      \right) \label{eq:suff:app:1:2}
    \end{eqnarray}
  } \normalsize
}
The first expression uses Cayleys' theorem and is adapted
from~\cite[(69)]{OYagan12}. The second and third expressions are
obtained by using bounds on $N_i.$ From~\eqref{eq:suff:app:1:1},
~\eqref{eq:suff:app:1:2}, if $c_{12}:=\left((1 + \delta) e
\right)^l/l^2,$ then we have
\begin{displaymath}
  \frac{1}{s_n^2} \sum_{l=2}^R {N_i \choose l} \prob{A_{N_i, l}} \leq
  \frac{c_{12} \left(R-1\right)}{n^{\frac{\frac{(1 - \delta) \alpha}{2
          \pi} - 1}{2} }}.
\end{displaymath}

\subsection{Deriving~\eqref{eq:app:suff:key:case2}}
\label{app:suff:key:case2}

In this section, we derive the probability of having isolated
component having sizes $< \ \min \left(N_i/ 2, P_n/ K_n \right).$
However there are an asymptotically large number of nodes, and thus
the number of keys shared is chosen to be $x = \lambda l K_n,$ where
$0 < \lambda < 1/2.$ Recall the expressions used for $s_n^2$ and $n
\pi r_n^2 K_n^2/P_n$ in Appendix~\ref{app:suff:key:case1}. 

\begin{enumerate}
  \item If $l \leq \lambda l K_n < N_i / 2,$ then:
    { \small {
        \begin{eqnarray*}
          && \hspace{-.4in} \frac{1}{s_n^2} \sum_{l = 1 +R}^{L_1(n)} {N_i
            \choose l} \prob{U_l \leq \lambda l K_n} \exp \left(-
            \left(N_i -l \right) \frac{K_n^2}{P_n} \right) \\
          && \hspace{-.7in} \leq \frac{1}{s_n^2} \sum_{l = 1 +R}^{L_1(n)}
          {N_i \choose \lambda l K_n} {P_n \choose \lambda l K_n}
          \left(\frac{\lambda l K_n}{P_n} \right)^{l K_n} \exp \left(-
            \left(N_i -l \right) \frac{K_n^2}{P_n} \right).
        \end{eqnarray*}
      }\normalsize 
    }
\item If $\lambda l K_n \leq l < P_n / 2,$ then:
    { \small {
        \begin{eqnarray*}
          && \hspace{-.4in} \frac{1}{s_n^2} \sum_{l = 1 +R}^{L_1(n)} {N_i
            \choose l} \prob{U_l \leq \lambda l K_n} \exp \left(-
            \left(N_i -l \right) \frac{K_n^2}{P_n} \right) \\
          && \hspace{-.4in} \leq \frac{1}{s_n^2} \sum_{l = 1 +R}^{L_1(n)}
          {N_i \choose l} {P_n \choose l}
          \left(\frac{\lambda l K_n}{P_n} \right)^{l K_n} \exp \left(-
            \left(N_i -l \right) \frac{K_n^2}{P_n} \right).
        \end{eqnarray*}
      }\normalsize 
    }
\end{enumerate}

\noindent Each term of the above sums (excluding the exponent) is
bounded above depending on $K_n$ as follows:
\begin{enumerate}
\item If $l < \lambda l K_n < N_i / 2,$ then
  { \small {
      \begin{eqnarray*}
        && \hspace{-.4in} {N_i \choose \lambda l K_n} {P_n \choose
          \lambda l K_n} \left(\frac{\lambda l K_n}{P_n} \right)^{l K_n} \\
        & \leq & \left(\frac{e N_i}{\lambda l K_n} \frac{e
            P_n}{\lambda l K_n} \right)^{\lambda l K_n} \left(\frac{\lambda l
            K_n}{P_n} \right)^{l K_n} \\
        & \leq & \left(\frac{e^2 (1 + \delta)}{\sigma} \right)^{l
          \lambda K_n} \left(\frac{\lambda l K_n}{P_n} \right)^{l K_n \left(1 -
            2 \lambda \right)} \\
        & = & \left( \left(\frac{e^2  (1 + \delta)}{\sigma} \right)^{
            \lambda} \lambda ^{\left(1 - 2 \lambda \right)} \right)^{K_n l}.
      \end{eqnarray*}
    }\normalsize }
  The first expression is from factorial bounds, while the second
  expression uses bounds on $N_i.$ The third expression is a
  rearrangement and it uses the fact that $l < L_1(n) \leq P_n / K_n.$
  Thus we have,
    { \small {
        \begin{eqnarray}
          && \hspace{-.3in} \frac{1}{s_n^2} \sum_{l = 1 +R}^{L_1(n)}
          {N_i \choose l} \prob{U_l \leq \lambda l K_n} \exp \left(-
            \left(N_i -l \right) \frac{K_n^2}{P_n} \right) \nonumber \\
          && \hspace{-.3in} \leq \frac{1}{s_n^2} \sum_{l = 1
            +R}^{L_1(n)} \left( \left(\frac{e^2 (1 + \delta)}{\sigma} \right)^{
              \lambda} \lambda ^{\left(1 - 2 \lambda \right)} \right)^{K_n l}
          \exponent{- \left(N_i -l \right) \frac{K_n^2}{P_n}} \nonumber \\
          && \hspace{-.3in} \leq \frac{\exponent{-N_i
              K_n^2/P_n}}{s_n^2} \sum_{l = 1 +R}^{L_1(n)} \left( e^{\frac{K_n}{P_n}}
            \left(\frac{e^2 (1 + \delta)}{\sigma} \right)^{ \lambda} \lambda
            ^{\left(1 - 2 \lambda \right)} \right)^{K_n l} \label{eq:app:suff:2:1}
        \end{eqnarray}
      }\normalsize 
    }
    \ignore{See that the final term raised to the $l$-th power is $<1$
      iff the term is $<1,$ so the necessary condition on $\sigma,
      \delta$ and $\lambda$ is
      \begin{displaymath}
        \left(\frac{e^2  (1 + \delta)}{\sigma} \right)^{ \lambda} \lambda
        ^{\left(1 - 2 \lambda \right)} < 1.
      \end{displaymath}
    }
  \item If $\lambda l K_n \leq l \leq P_n/2,$ then 
    { 
      \small {
        \begin{eqnarray*}
          && \hspace{-.4in} {N_i \choose l} {P_n \choose l}
          \left(\frac{\lambda l K_n}{P_n} \right)^{l K_n} \\
          & \leq & \left( \frac{e N_i}{l} \right)^{l} \left(
            \frac{e P_n}{l} \right)^{l} \left( \frac{\lambda l K_n}{P_n}
          \right)^{lK_n} \\
          & \leq & \left( \frac{e^2 (1 + \delta)}{ \sigma} \left(
              \frac{l}{P_n} \right)^{K_n - 2} \right)^l\\
          & \leq & \left( \frac{e^2 (1 + \delta)}{ 2^{K_n-2} \sigma}
          \right)^l.
        \end{eqnarray*}
      }\normalsize }
      Thus we have,
    { \small {
        \begin{eqnarray}
          && \hspace{-.3in} \frac{1}{s_n^2} \sum_{l = 1 +R}^{L_1(n)} {N_i
            \choose l} \prob{U_l \leq \lambda l K_n} \exp \left(-
            \left(N_i -l \right) \frac{K_n^2}{P_n} \right) \nonumber \\
          && \hspace{-.3in} \leq \frac{1}{s_n^2} \sum_{l = 1
            +R}^{L_1(n)} \left( \frac{e^2 (1 + \delta)}{ 2^{K_n-2} \sigma}
          \right)^l\exponent{- \left(N_i -l \right) \frac{K_n^2}{P_n}} \nonumber
          \\
          && \hspace{-.3in} \leq \frac{\exponent{-N_i
              K_n^2/P_n}}{s_n^2} \sum_{l = 1 +R}^{L_1(n)} \left( \frac{e^{2 +
                \frac{K_n^2}{P_n}} (1 + \delta)}{2^{K_n - 2} \sigma} \right)^{
            l} \label{eq:app:suff:2:2}
        \end{eqnarray}
      }\normalsize 
    }

    \ignore{ For the term raised to the $l$-th power to be $< 1,$ the
      condition on $\sigma$ is that
      \begin{displaymath}
        e^2 < 2^{K_n -2} \sigma.
      \end{displaymath}
    }
\end{enumerate}
If
{\small
  {
    \begin{displaymath}
      \max \left(\frac{e^{2 + \frac{K_n^2}{P_n}} (1 + \delta)}{2^{K_n
            - 2} \sigma}, e^{\frac{K_n}{P_n}} \left(\frac{e^2 (1 +
            \delta)}{\sigma} \right)^{ \lambda} \lambda ^{\left(1 - 2 \lambda
          \right)} \right) < 1
    \end{displaymath}
  }\normalsize
}
then the sum terms in~\eqref{eq:app:suff:2:1}
and~\eqref{eq:app:suff:2:2} form a geometric series which sums to a
finite number. And we know that
\begin{eqnarray*}
  && \frac{\exponent{-N_i K_n^2/P_n}}{s_n^2} \\
  && \leq \exponent{- (1 - \delta) n s_n^2 K_n^2/ P_n - \log (s_n^2)}
  \\
  && \leq \exponent{- \frac{(1 - \delta) \alpha}{2 \pi} \log n - \log
    (d_n/2\pi) + \log n }.
\end{eqnarray*}
Now since $\alpha > 2 \pi / (1 - \delta),$ the above expression $\to
0.$ Thus giving us the necessary bound. Now consider the second term
from~\eqref{eq:10.2}.
{\small{
    \begin{eqnarray*}
      && \frac{1}{s_n^2} \sum_{l = R+1}^{L_1(n)} {N_i \choose l}
      \prob{C_l} \exponent{- \left(N_i - l \right) \frac{K_n}{P_n}
        \left(1 + \lambda l K_n \right)} \\
      && \leq \sum_{l = R+1}^{L_1(n)} \frac{1}{s_n^2} \left(\frac{N_i
          e}{l} \right)^l \ l^{l-2} \beta_n^{l-1} \exponent{- \left(N_i - l
        \right) \frac{\lambda l K_n^2}{P_n}} \\
      && \leq \sum_{l = R+1}^{L_1(n)} \frac{ \left( (1 + \delta) e
        \right)^l }{l^2} \ n . \left(n s_n^2 \beta_n \right)^{l-1} \exponent{-
        \left(N_i - l \right) \frac{\lambda l K_n^2}{P_n}} \\
      && = \sum_{l = R+1}^{L_1(n)} \frac{ \left( (1 + \delta) e
        \right)^l }{l^2} \ n . \left(\frac{\alpha}{2 \pi} \log n \right)^{l-1}
      \\
      && \hspace{.2in} \times \exponent{- \left( 1 - \delta -
          \frac{l}{n s_n^2} \right) \lambda l \left( n s_n^2\frac{K_n^2}{P_n}
        \right)} \\
      && = \sum_{l = R+1}^{L_1(n)} \frac{ \left( (1 + \delta) e
        \right)^l }{l^2} \ n . \left(\frac{\alpha}{2 \pi} \log n \right)^{l-1}
      \\
      && \hspace{.2in} \times \exponent{- \left( 1 - \delta -
          \frac{l}{n s_n^2} \right) \lambda l \left( \frac{\alpha}{2 \pi}
        \right) \log n } \\
      && = \sum_{l = R+1}^{L_1(n)} \left(\frac{\left(1 + \delta
          \right) e \alpha}{2 \pi}\right)^l \\
      && \times \exponent{- l \log n \left(\left(1 - \delta -
            \frac{l}{ns_n^2} \right) \left(\frac{\alpha \lambda}{2 \pi}\right) -
          \frac{1}{l} - \frac{\log \log n}{\log n} \right)}
    \end{eqnarray*}
  }\normalsize
}
\noindent The first inequality uses bounds on the factorial,
$\prob{C_l}$ and $\exponent{-(N_i - l)K_n/P_n } \leq 1.$ The second
inequality uses bounds on $N_i.$ Consider the term multiplying $l \log
n$ in the exponent of the preceding expression. The following two
conditions are possible:
\begin{enumerate}
\item If $l$ is large ($\to \infty$ as a function of $n$) then it is
  easy to see that this term is positive.
\item If $l$ is a constant, then using the condition that $\lambda R >
  2 \pi / \left(\alpha \left(1 - \delta\right)\right),$ the term is
  positive.
\end{enumerate}
From both the above conditions, the preceding expression can be
written as a sum of a geometric series of the form $\sum_{l =
  R_1}^{L_1(n)} \eta^l,$ with $\eta \to 0.$ And so
{\small{
    \begin{eqnarray*}
      && \frac{1}{s_n^2} \sum_{l = R+1}^{L_1(n)} {N_i \choose l}
      \prob{C_l} \exponent{- \left(N_i - l \right) \frac{K_n}{P_n}
        \left(1 + \lambda l K_n \right)} \\
      && < \sum_{l > R} \eta^l = \frac{\eta^{R}}{1 - \eta} =:
      \frac{c_6}{n^{c_7}} \to 0
    \end{eqnarray*}
  }\normalsize
}
where $c_6, c_7$ are appropriately chosen.

\subsection{Deriving~\eqref{eq:app:suff:key:case3}}
\label{app:suff:key:case3}

In this subsection, we use $x = \mu P_n$ as in~\cite{OYagan12} with $0
< \mu < 0.44.$ The expressions used for $s_n^2$ and $n \pi r_n^2
K_n^2/P_n$ in Appendix~\ref{app:suff:key:case1} hold.
{ \small{
    \begin{eqnarray*}
      && \hspace{-.3in} \frac{1}{s_n^2} \sum_{l=L_1(n) + 1}^{N_i/2}
      {N_i \choose l} \prob{U_{l} \leq \mu P_n}
      \exponent{-\left(N_i - l \right) \frac{K_n^2}{P_n}} \\
      & \leq & \frac{1}{s_n^2} \sum_{l=L_1(n) + 1}^{N_i/2} {N_i
        \choose l} {P_n \choose \mu P_n} \ \mu^{l K_n} \exponent{-\left(N_i -
          l \right) \frac{K_n^2}{P_n}} 
    \end{eqnarray*}
    \begin{eqnarray*}
      & \leq & \frac{\exponent{-\frac{N_i K_n^2}{2P_n}}}{s_n^2}
      \sum_{l=L_1(n) + 1}^{N_i/2} {N_i \choose l}
      \left(\frac{e}{\mu}\right)^{\mu P_n} \ \mu^{l K_n}.
    \end{eqnarray*}
  } \normalsize
}

\noindent The second step uses the bounds on the factorial. Using this
inequality and $\sum_{l=L_1(n) + 1}^{N_i/2} {N_i \choose l} \leq
2^{N_i}$ in the above and using $P_n \geq \sigma n s_n^2$ for $\sigma
> 0,$ we have
{\small{
    \begin{eqnarray*}
      && \hspace{-.2in} \frac{\exponent{-\frac{N_i
            K_n^2}{2P_n}}}{s_n^2} \sum_{l=L_1(n) + 1}^{N_i/2} {N_i \choose l}
      \left(\frac{e}{\mu}\right)^{\mu P_n} \ \mu^{l K_n} \\
      && \hspace{-.2in} \leq \exponent{-\frac{N_i K_n^2}{2P_n} - \log s_n^2 + \mu P_n
        \log \left( \frac{e}{\mu} \right) + P_n \log \mu + N_i \log 2 } \\
      && \hspace{-.2in} \leq \exp \left( - \frac{\left(1 -
            \delta\right) \alpha }{4 \pi } \log n - \log s_n^2 \right. \\
      && \left. - P_n \left( \log \left( \frac{1}{\mu} \right) + \mu
          \log \left(\frac{e}{\mu} \right) \right) + (1 + \delta) n s_n^2 \log 2
      \right) \\
      && \hspace{-.2in} \leq \exp \left( - \frac{\left(1 -
            \delta\right) \alpha }{4 \pi } \log n - \log s_n^2 \right. \\
      && \left. - n s_n^2 \left( \sigma \log
          \left(\frac{e^{\mu}}{\mu^{1 + \mu}} \right) - (1 + \delta) \log
          2\right) \right)\\
      && \hspace{-.2in} = \exp \left( - \frac{\left(1 -
            \delta\right) \alpha }{4 \pi } \log n - \log \left(\frac{d_n}{2 \pi n}
        \right) \right. \\
      && \left. - \frac{d_n}{2\pi} \left( \sigma \log
          \left(\frac{e^{\mu}}{\mu^{1 + \mu}} \right) - (1 + \delta) \log
          2\right) \right) \\
      && \hspace{-.2in} = \exp \left( - d_n \left( \frac{1}{2\pi}
          \left( \sigma \log \left(\frac{e^{\mu}}{\mu^{1 + \mu}} \right) - (1 +
            \delta) \log 2\right) \right. \right. \\
      && \left. \left. + \frac{(1 - \delta)\alpha}{4 \pi} \frac{\log
            n}{d_n} + \frac{\log d_n}{d_n} - \frac{\log (2 \pi n)}{d_n} \right)
      \right).
    \end{eqnarray*}
  }\normalsize
}

\noindent The first inequality above uses $\mu < 1$ and the remaining
steps are direct. The above bounds are direct and the final expression
$\to 0$ since
\begin{displaymath}
  \sigma > \frac{ (1 + \delta) \log 2 }{ \log \left(
      \frac{e^\mu}{\mu^{1 + \mu} } \right) }.
\end{displaymath}
Now consider the second expression from~\eqref{eq:10.2}.
{ \small{
    \begin{eqnarray*}
      && \hspace{-.2in} \frac{1}{s_n^2} \sum_{l = 1 + L_1(n)}^{N_i/2}
      {N_i \choose l} \prob{C_l} \exp \left( - \left(N_i -l \right)
        \frac{K_n}{P_n} \left( 1 + \mu P_n \right) \right) \\
      & \leq & \sum_{l = 1 + L_1(n)}^{N_i/2} {N_i \choose l}
      \exponent{ - \frac{N_i}{2} K_n \mu - \log s_n^2 } \\
      & \leq & \exponent{ - \frac{N_i}{2} K_n \mu - \log s_n^2 + N_i
        \log 2} \\
      & = & \exponent{ - n s_n^2 (1 - \delta) \left(\frac{\mu K_n}{2}
          - \log 2 \right) -\log \left( \frac{d_n}{2 \pi} \right) + \log n} \\
      & = & \exponent{ - \frac{d_n}{2\pi} (1 - \delta) \left(\frac{\mu
            K_n}{2} - \log 2 \right) -\log \left( \frac{d_n}{2 \pi} \right) + \log
        n}\\
      & = & \exponent{ -d_n \left(\frac{1 - \delta}{2\pi}
          \left(\frac{\mu K_n}{2} - \log 2 \right) -\frac{\log \left(
              \frac{d_n}{2 \pi} \right)}{d_n} + \frac{\log n}{d_n} \right)}
    \end{eqnarray*}
  }\normalsize
}

\noindent The first inequality uses $\prob{C_l} \leq 1$ and
$\exponent{-N_i K_n/ 2 P_n} \leq 1.$ The second inequality uses
$\sum_{l = 1 + L_1(n)}^{N_i/2} {N_i \choose l} \leq 2^{N_i}.$ For any
$K > 2 \log 2/\mu,$ the above probability $\to 0.$ 

\end{document}